\documentclass[a4paper,11pt]{article}

\usepackage{jheppub,amsmath,amssymb,graphicx,mathrsfs,amsopn,stmaryrd,amsthm,color,bm} 
\usepackage[mathscr]{euscript}
\usepackage[final]{showlabels}
\usepackage[T1]{fontenc} 
\numberwithin{equation}{section}
\newtheorem{thm}{Theorem}[section]
\newtheorem{prop}[thm]{Proposition}
\newtheorem{lem}[thm]{Lemma}
\newtheorem{cor}[thm]{Corollary}

\newtheorem{eg}[thm]{Example}

\newtheorem{rem}[thm]{Remark}
\allowdisplaybreaks
\newcommand{\beq}{\begin{equation}}
\newcommand{\eeq}{\end{equation}}
\newcommand{\be}{\begin{equation*}}
\newcommand{\ee}{\end{equation*}}

\newcommand{\C}{\mathbb{C}}

\newcommand{\Z}{\mathbb{Z}}

\newcommand{\bB}{\mathbb{B}}

\newcommand{\bA}{\mathbb{A}}

\newcommand{\mc}{\mathcal}

\newcommand{\cN}{\mathcal{N}}

\newcommand{\gl}{\mathfrak{gl}}
\newcommand{\h}{\mathfrak{h}}
\newcommand{\n}{\mathfrak{n}}

\newcommand{\fkS}{\mathfrak{S}}

\newcommand{\End}{\mathrm{End}}

\newcommand{\id}{{\mathrm{id}}}

\newcommand{\str}{{\mathrm{str}}} 

\newcommand{\scrV}{\mathscr{V}}
\newcommand{\sN}{\mathscr{N}}
\newcommand{\sW}{\mathscr{W}}
\newcommand{\sV}{\mathscr{V}}

\newcommand{\pa}{\partial}
\newcommand{\tl}{\tilde}
\newcommand{\gge}{\geqslant}
\newcommand{\lle}{\leqslant}
\newcommand{\la}{\lambda}
\newcommand{\La}{\Lambda}

\newcommand{\bla}{\bm\lambda}
\newcommand{\bLa}{\bm\Lambda}
\newcommand{\glMN}{\mathfrak{gl}_{m|n}}
\newcommand{\glNM}{\mathfrak{gl}_{n|m}}
\newcommand{\UglMN}{\mathrm{U}(\mathfrak{gl}_{m|n})}

\newcommand{\YglMN}{\mathrm{Y}(\mathfrak{gl}_{m|n})}

\newcommand{\YglNM}{\mathrm{Y}(\mathfrak{gl}_{n|m})}
\newcommand{\wt}{\widehat}    
\newcommand{\ka}{\kappa}
\newcommand{\ve}{\varepsilon}

\newcommand{\qedd}{\tag*{$\square$}}

\newcommand{\EndV}{\End(\scrV)}

\newcommand{\bfv}{\mathbf{v}}
\newcommand{\bfw}{\mathbf{w}}
\newcommand{\bfx}{\mathbf{x}}
\newcommand{\YMNplus}{\mathrm{Y}_+(\glMN)}
\newcommand{\llp}{\llparenthesis}
\newcommand{\rrp}{\rrparenthesis}

\title{\boldmath On Bethe eigenvectors and higher transfer matrices for supersymmetric spin chains}

\author[a,1]{Kang Lu,\note{Corresponding author.}}

\affiliation[a]{Department of Mathematics, University of Virginia,\\141 Cabell Dr, Charlottesville, VA 22903, USA}

\emailAdd{Kang.Lu@virginia.edu}

\abstract{We study the $\gl_{m|n}$ XXX spin chains defined on tensor products of highest weight $\gl_{m|n}$-modules. We show that the on-shell Bethe vectors are eigenvectors of higher transfer matrices and compute the corresponding eigenvalues. Then we take the classical limits and obtain the corresponding results for the $\gl_{m|n}$ Gaudin models.}
\begin{document} 
\maketitle


		
		



\section{Introduction}
In this paper, we study the quantum integrable models associated with the Lie superalgebra $\glMN$, the XXX spin chains and Gaudin models. An important problem in the study of quantum integrable systems is to find the eigenvectors and the corresponding eigenvalues of the Hamiltonians. 

Let $m,n\in\Z_{\gge 0}$ and set $N=m+n$. Let $T_{ab}(u)$, $1\lle a,b\lle N$, be the RTT generating series of the super Yangian $\YglMN$ associated with $\glMN$ and $T(u)=\sum_{a,b=1}^N E_{ab}\otimes T_{ab}(u)$ the monodromy matrix. The integral of motions for XXX spin chains are given by the quantum Berezinian (super-determinant) of a certain Manin matrix, see \cite{MR14},
\[
\mathrm{Ber}(1-T^\dag(u)Q e^{-\pa_u})=1+\sum_{k=1}^\infty (-1)^k \mathscr T_{k,Q}(u)e^{-k\pa_u},
\]
where $Q=\mathrm{diag}(Q_1,\dots,Q_{N})$ is invertible and $\dag:E_{ab}\mapsto (-1)^{|a||b|+|a|}E_{ba}$ is the super transpose. The series $\mathscr T_{k,Q}(u)$, whose coefficients are elements in the super Yangian $\YglMN$, are called {\it transfer matrices}. The subalgebra of $\YglMN$ generated by the coefficients of $\mathscr T_{k,Q}(u)$ for all $k>0$ is commutative and called the Bethe subalgebra, cf. \cite{NO}.

Given a highest weight representation of $\YglMN$, we are interested in finding the joint eigenvectors of transfer matrices. The same problem for diagonalizing the standard transfer matrix $\mathscr T_{1,Q}(u)$ has been studied in \cite{BR08} using algebraic Bethe ansatz \cite{STF79,TF79} and nested Bethe ansatz \cite{KR83}. The main result of \cite{BR08} indicates that if the sequence of parameters $\bm t$ satisfies the Bethe ansatz equation, then one can construct a Bethe vector $\bB(\bm t)$ which (if nonzero) is an eigenvector of $\mathscr T_{1,Q}(u)$. Moreover, the corresponding eigenvalue can be computed explicitly. Following \cite{MTV06}, we extend the results of \cite{BR08} to show that the Bethe vector $\bB(\bm t)$ is indeed a joint eigenvector of all transfer matrices $\mathscr T_{k,Q}(u)$, $k\in\Z_{>0}$, and give the explicit eigenvalues for each transfer matrix, cf. \cite{Tsu97}. 

Let me explain our results in more detail. Let $M$ be a highest weight representation of $\YglMN$ generated by a nonzero vector $v$ such that $T_{aa}(u)v=\zeta_a(u)v$ and $T_{ab}(u)v=0$ for all $1\lle b<a\lle N$, where $\zeta_a(u)$ is an arbitrary power series of $u^{-1}$ in $1+u^{-1}\C[[u^{-1}]]$. Let $\bm \xi=(\xi^1,\dots,\xi^{N-1})$ be a sequence of nonnegative integers and $\bm t=(t_1^1,\dots,t_{\xi^1}^1;\dots;t_1^{N-1},\dots,t_{\xi^{N-1}}^{N-1})$ a sequence of complex numbers. Set $y_a(u)=\prod_{i=1}^{\xi^a}(u-t_i^a)$, $1\lle a<N$, and $y_0(u)=y_{N}(u)=1$. Let $\ka_a=1$ if $1\lle a\lle m$ and $\ka_a=-1$ if $m< a\lle N$. The Bethe ansatz equation for XXX spin chain associated with $\glMN$ is a system of algebraic equations in $\bm t$,
\[
-\frac{\ka_aQ_a}{\ka_{a+1}Q_{a+1}}\frac{\zeta_a(t_i^a)}{\zeta_{a+1}(t_i^a)}\frac{y_{a-1}(t_i^a+\ka_a)}{y_{a-1}(t_i^a)}\frac{y_a(t_i^a-\ka_a)}{y_a(t_i^a+\ka_{a+1})}\frac{y_{a+1}(t_i^a)}{y_{a+1}(t_i^a-\ka_{a+1})}=1,
\]
for $1\lle a<N$ and $1\lle i\lle \xi^a$, see \eqref{eq:bae}. Suppose $\bm t$ satisfies the Bethe ansatz equation and let $\bB_{\bm \xi}^v(\bm t)\in M$ be the corresponding (on-shell) Bethe vector, see e.g.  \eqref{eq:bae-sym}, \eqref{eq:off-shell-bv}. Then the eigenvalues of the transfer matrices acting on $\bB^v_{\bm \xi}(\bm t)$ can be compactly described as follows,
\[
\mathrm{Ber}(1-QT(u)e^{-\pa_u})\bB^v_{\bm \xi}(\bm t)=\bB^v_{\bm \xi}(\bm t)\mathop{\overrightarrow\prod}\limits_{1\lle a\lle N}\Big(1-Q_a\zeta_a(u)\frac{y_{a-1}(u+\ka_a)y_a(u-\ka_a)}{y_{a-1}(u)y_a(u)}e^{-\pa_u}\Big)^{\ka_a},
\]
see Corollaries \ref{cor:main-tech}, \ref{cor:main-tech-xxx}, which follows from the main technical results, Theorem \ref{thm:main-tech}. Such a statement was previously established for the case $n=0$ in \cite{MTV06}, and for the case $m=n=1$ in \cite[Theorem 6.4]{LM21a} by a brute force computation.

We show it for the case when $M$ is a tensor product of evaluation highest weight modules. Such a rational difference operator also appeared in \cite{Tsu97} when $M$ is a tensor product of evaluation vector representations. However, the Bethe vector is not discussed there. Note that the same approach also works for the full generality.

We prove Theorem \ref{thm:main-tech} by induction on $m$. For the base case of $\gl_{0|n}$, it is essentially \cite[Theorem 5.2]{MTV06} by the correspondences between transfer matrices and Bethe vectors for $\YglMN$ and $\YglNM$, see Section \ref{sec:correspondence}. Then we perform the nested Bethe ansatz and use induction hypothesis. Since our induction is on $m$, the first index is always even. Therefore the procedure for nested Bethe ansatz turns out to be similar to the case of the nonsuper case as in \cite{MTV06}.

By taking the classical limits, we obtain the corresponding statement for the Gaudin models associated with $\glMN$. Explicitly, let $M_1,\dots, M_\ell$ be highest weight $\glMN$-modules with highest weights $\La_1,\dots,\La_\ell$, where $\La_i=(\La_i^1,\dots,\La_i^N)$.  Let $\bm z=(z_1,\dots,z_\ell)$ be a sequence of pairwise distinct complex numbers. Let $\bm \xi$, $\bm t$, and $y_a(u)$, $0\lle a\lle N$, be  as before. The Bethe ansatz equation for Gaudin models associated with $\glMN$ is a system of algebraic equations in $\bm t$, 
\[
\begin{split}
K_a-K_{a+1}+\sum_{j=1}^\ell\frac{\ka_a\La_j^a-\ka_{a+1}\La_j^{a+1}}{t_i^a-z_j}+& \,\frac{\ka_ay_{a-1}'(t_i^a)}{y_{a-1}(t_i^a)}\\-&\, \frac{(\ka_a+\ka_{a+1})y_{a}''(t_i^a)}{2y_a'(t_i^a)}+\frac{\ka_{a+1}y_{a+1}'(t_i^a)}{y_{a+1}(t_i^a)}=0,
\end{split}
\]
for $1\lle a< N$ and $1\lle i\lle \xi^a$, see \eqref{eq:Gaudinbae}. Suppose $\bm t$ satisfies the Bethe ansatz equation and let $\mathbb F^v_{\bm \xi}(\bm t)\in M_1\otimes\cdots\otimes M_\ell$ be the corresponding (on-shell) Bethe vector. The Gaudin transfer matrices (higher Gaudin Hamiltonians) are elements of the universal enveloping superalgebra $\mathrm{U}(\gl_{m|n}[x])$ of the current superalgebra $\gl_{m|n}[x]$ which are again given by the quantum Berezinian, see \cite{MR14,HM20},
\[
\mathrm{Ber}(\pa_u-K-L^\dag(u))=\sum_{r=0}^{\infty}\mathcal G_{r,K}(u)\pa_u^{m-n-r},
\]
where $L(u)$ is the generating matrix of $\gl_{m|n}[x]$ and $K=\mathrm{diag}(K_1,\dots,K_N)$. Then the eigenvalues of the  Gaudin transfer matrices acting on $\mathbb F^v_{\bm \xi}(\bm t)$ can be compactly described as follows,
\[
\begin{split}
\mathrm{Ber}(\pa_u-K-L^\dag(u))&\,\mathbb F^v_{\bm \xi}(\bm t)\\=&\,\mathbb F^v_{\bm \xi}(\bm t)\mathop{\overrightarrow\prod}\limits_{1\lle a\lle N}\bigg(\pa_u-K_a-\ka_a\Big(\sum_{j=1}^\ell \frac{\La_j^a}{u-z_j}+\frac{y_{a-1}'(u)}{y_{a-1}(u)}-\frac{y_{a}'(u)}{y_{a}(u)}\Big)\bigg)^{\ka_a},
\end{split}
\]
see Theorem \ref{thm:main-tech-Gaudin}.

Finally, we remark that the results can be generalized beyond the standard parity sequence (root system). Indeed, for a super Yangian of type A whose first index is odd, in \eqref{eq:high-R} and \eqref{eq:T-decom}, one can use symmetric power instead of anti-symmetric power. Alternatively, one can introduce the shift parameter for the super Yangian. Changing parity of the fundamental space is related to negating the shift parameter.  An interesting question is as follows: given a finite-dimensional irreducible representation $M$ of $\YglMN$, how a Bethe vector associated with a solution $\bm t$ corresponds to the Bethe vector associated with another solution $\tilde{\bm t}$ obtained from $\bm t$ by the fermionic reproduction procedure (odd reflection), see e.g. \cite{HMVY19,HLM19,Mol21,Lu21}. When the twisting matrix $Q$ is regular semisimple, one would expect that they are proportional. When the representation is the tensor product of evaluation polynomial modules and $Q$ is the identity matrix, they are related by a simple Lie superalgebra action, see \cite[Corollary 5.6]{HMVY19}. For more general cases, the relation remains unclear, see a related discussion in \cite[Section 8.3]{HMVY19}. It is also well known that in general Bethe eigenvectors obtained from this approach does not provide the full list of eigenvectors of the Bethe subalgebra. A complete spectrum of the Bethe subalgebra should be described using the (extended) Q-systems, see \cite{MV17,HLM19,GJS22}. Results on the completeness of Bethe ansatz for supersymmetric XXX spin chains can be found in \cite{LM21a,CLV22}. A conjectural complete spectrum of supersymmetric XXX spin chains is proposed in \cite{MNV} using separation of variable bases, see also \cite{Ryan} for a recent review of separation of variables. 

The paper is a step toward understanding the perfect duality of supersymmetric spin chains in the sense of \cite{Lu20}, see \cite{CLV22} for the case when the underlying Hilbert space is a tensor product of evaluation vector representations. The corresponding statements for the even case were used to prove the perfect duality for Gaudin models and XXX spin chains in \cite{MTV09,MTV14}.

The paper is organized as follows. We recall the basics of the Lie superalgebra $\glMN$, the corresponding Yangian, and the fusion procedure in Section \ref{sec:pre}. Section \ref{sec:XXX} is devoted to the study of XXX spin chains where the main results for XXX spin chains are given. We prove the main technical result Theorem \ref{thm:main-tech} in Section \ref{sec:more}. We discuss Gaudin models and the corresponding main results in Section \ref{sec:Gaudin}. By taking the classical limits, we prove the main results for Gaudin models in Section \ref{sec:Gaudin-proof}.

\medskip

{\bf Acknowledgement.} The author thanks E. Mukhin and V. Tarasov for valuable discussions.

\section{Preliminaries}\label{sec:pre}

\subsection{Basics}
Throughout this article we will use the following general conventions. 

A \emph{vector superspace} $W = W_{\bar 0}\oplus W_{\bar 1}$ is a $\Z_2$-graded vector space. We call elements of $W_{\bar 0}$ \emph{even} and elements of
$W_{\bar 1}$ \emph{odd}. We write $|w|\in\{\bar 0,\bar 1\}$ for the parity of a homogeneous element $w\in W$. Whenever $|v|$ is used, we implicitly assume that $v$ is homogeneous. Set $(-1)^{\bar 0}=1$ and $(-1)^{\bar 1}=-1$. 

Let $A$ and $B$ be associative superalgebras (namely $\Z_2$-graded algebras). Then their tensor product $A\otimes B$ is also an associative superalgebra with the structure given by
$$
(a\otimes b)(a'\otimes b')=aa'\otimes bb'(-1)^{|a'||b|}, \quad |a\otimes b| = |a|+|b|,
$$
for any elements $a,a'\in A$ and $b,b'\in B$. 

For any $\Z_2$-graded modules $V$ and $W$ over $A$ and $B$, respectively, the vector superspace $V\otimes W$ is a $\Z_2$-graded module over $A\otimes B$ with the structure given by
$$
(a\otimes b)(v\otimes w)=av\otimes bw(-1)^{|b||v|}, \quad |v\otimes w|=|v|+|w|
$$
for any elements $v\in V$ and $w\in W$.

A superalgebra \emph{homomorphism} $\alpha:A\to B$ is a linear map satisfying $\alpha(aa')=\alpha(a)\alpha(a')$ for all $a,a'\in A$. A superalgebra \emph{antihomomorphism} $\beta:A\to B$ is a linear map satisfying $\beta(aa')=\beta(a')\beta(a)(-1)^{|a||a'|}$ for all $a,a'\in A$. 


We use the standard superscript notation for embeddings of tensor factors into tensor products. If $A_1, \dots$, $A_k$ are unital associative superalgebras, and $a\in A_i$, then
\[
a^{(i)}=1^{\otimes(i-1)}\otimes a\otimes 1^{\otimes(k-i)}\in A_1\otimes \cdots\otimes A_k.
\]If $a\in A_i$ and $b\in A_j$, then $(a\otimes b)^{(ij)}=a^{(i)}b^{(j)}$, etc. For example, if $k=2$ and $A_1=A_2=A$, then 
$$
a^{(1)}=a\otimes 1, \quad b^{(2)}=1\otimes b,\quad (a\otimes b)^{(12)}=a\otimes b,\quad (a\otimes b)^{(21)}=b\otimes a(-1)^{|a||b|}.
$$
For products of noncommuting factors, we use the following notation:
\begin{align*}
&\mathop{\overrightarrow\prod}\limits_{1\lle i\lle k}X_i=X_1\cdots X_k, &\mathop{\overleftarrow\prod}\limits_{1\lle i\lle k}X_i=X_k\cdots X_1,	\\
&\mathop{\overrightarrow\prod}\limits_{1\lle i<j\lle k} =\mathop{\overrightarrow\prod}\limits_{1\lle i\lle k}\ \mathop{\overrightarrow\prod}\limits_{i<j\lle k}, &\mathop{\overleftarrow\prod}\limits_{1\lle i<j\lle k} =\mathop{\overleftarrow\prod}\limits_{1\lle j\lle k}\ \mathop{\overleftarrow\prod}\limits_{1\lle i<j}.
\end{align*}

\subsection{Lie superalgebra $\glMN$}Fix $m,n\in\Z_{\gge 0}$. Set $N=m+n$ and $\mathscr N=m-n$. Let $\C^{m|n}$ be a complex vector superspace, with $\dim(\C^{m|n})_{\bar 0}=m$ and $\dim(\C^{m|n})_{\bar 1}=n$. Note that $\dim(\C^{m|n})=N$ and $\mathrm{sdim}(\C^{m|n})=\mathscr N$. Choose a homogeneous basis $\bfv_a$, $1\lle a\lle N$, of $\C^{m|n}$ such that $|\bfv_a|=\bar 0$ for $1\lle a\lle m$ and $|\bfv_a|=\bar 1$ for $m< a\lle N$. We call it the \emph{standard basis} of $\C^{m|n}$. Set $|a|=|\bfv_a|$. 

Denote $\C^{m|n}$ by $\scrV$ and consider elements of $\EndV$ as (super)matrices with respect to the standard basis of $\scrV$. In particular, we have the matrix units $E_{ab}$ such that $E_{ab}\bfv_{c}=\delta_{bc}\bfv_a$ for $1\lle a,b,c\lle N$.

The Lie superalgebra $\glMN$ is generated by elements $e_{ab}$, $1\lle a,b\lle N$, with the supercommutator relations
\beq\label{eq:nilpotent}
[e_{ab},e_{cd}]=\delta_{bc}e_{ad}-(-1)^{(|a|+|b|)(|c|+|d|)}\delta_{ad}e_{cb},
\eeq
and the parity of $e_{ab}$ is given by $|a|+|b|$. We have the standard nilpotent subalgebras of $\glMN$,
\[
\n_+=\bigoplus_{1\lle a< b\lle N}\C e_{ab},\quad \n_-=\bigoplus_{1\lle a< b\lle N}\C e_{ba}.
\]

A vector $v$ in a $\glMN$-module is a \emph{weight vector of weight} $(\La^1,\dots,\La^N)$ if $e_{aa}v=\La^a v$ for $1\lle a\lle N$. A vector $v$ is called a \emph{singular vector} if $e_{ab}v=0$ for any $1\lle a<b\lle N$.

Let $P=\sum_{a,b=1}^N E_{ab}\otimes E_{ba}(-1)^{|b|}$ which is the super flip operator: $P(v\otimes w)=(-1)^{|v||w|}w\otimes v$. Clearly, we have $X^{(21)}=PXP$ for $X\in \End(\scrV^{\otimes 2})$.

Let $S_k$ be the symmetric group permuting $\{1,2,\dots,k\}$ with the simple permutation $\sigma_i=(i,i+1)$ for $1\lle i< k$. The symmetric group acts on $\scrV^{\otimes k}$ by $\sigma_i\mapsto P^{(i,i+1)}$. Let $\mathbb H_{\{k\}}$, $\bA_{\{k\}}\in\End(\scrV^{\otimes k})$ be the symmetrizer and anti-symmetrizer, respectively,
\[
\mathbb H_{\{k\}} =\frac{1}{k!}\sum_{\sigma\in S_k} \sigma,\quad \bA_{\{k\}} =\frac{1}{k!}\sum_{\sigma\in S_k} \mathrm{sign}(\sigma)\sigma,
\]where $\sigma$ is identified as the corresponding operator in $\End(\scrV^{\otimes k})$. Denote by $\scrV^{\wedge k}$ the image of $\bA_{\{k\}}$. For any even matrix $Q\in \EndV$, set $Q^{\wedge k}={Q^{\otimes k}}\big|_{\scrV^{\wedge k}}$.

The map $\glMN\to \EndV,\ e_{ab}\mapsto E_{ab}$ defines a $\glMN$-module structure on $\scrV$. We call it the \emph{vector representation} of $\glMN$. The space $\scrV^{\wedge k}$ is also a $\glMN$-module with the action $e_{ab}\mapsto (E_{ab}^{(1)}+\cdots+E_{ab}^{(k)})\big|_{\scrV^{\wedge k}}$.

Define a \emph{supertrace} $\str:\EndV \to \C$, which is supercyclic,
\beq\label{eq:cyc-tr}
\str(E_{ab})=(-1)^{|b|}\delta_{ab},\qquad \str([E_{ab},E_{cd}])=0.
\eeq
Define the supertranspose $\dag$,
\[
\dag:\EndV\to \EndV,\quad E_{ab}\mapsto (-1)^{|a||b|+|a|}E_{ba}.
\]
The supertranspose is an anti-homomorphism and respects the supertrace,
\beq\label{eq transpose cyclic}
(AB)^\dag=(-1)^{|A||B|}B^\dag A^\dag,\qquad \str(A)=\str(A^\dag), 
\eeq
for all supermatrices $A$, $B$.

\subsection{Super Yangian}
Define the \emph{rational R-matrix} $R(u)=u+P\in \End(\scrV^{\otimes 2})$ which satisfies the Yang-Baxter equation
\beq\label{eq:Yang-Baxter}
R^{(12)}(u-v)R^{(13)}(u)R^{(23)}(v)=R^{(23)}(v)R^{(13)}(u)R^{(12)}(u-v).
\eeq

The \emph{super Yangian} $\YglMN$ is a unital associative superalgebra with generators $T_{ab}^{\{s\}}$ of parity $|a|+|b|$, $1\lle a,b\lle N$, $s\in \Z_{>0}$. Consider the generating series
\[
T_{ab}(u)=\delta_{ab}+\sum_{s=1}^\infty T_{ab}^{\{s\}}u^{-s}
\]
and combine the series into a linear operator $T(u)=\sum_{a,b=1}^{N}E_{ab}\otimes T_{ab}(u)\in \End(\scrV)\otimes \YglMN[[u^{-1}]]$. The defining relations of $\YglMN$ are given by
\beq\label{eq RTT}
R^{(12)}(u-v)T^{(13)}(u)T^{(23)}(v)=T^{(23)}(v)T^{(13)}(u)R^{(12)}(u-v).
\eeq
Alternatively, the defining relation \eqref{eq RTT} gives
\beq\label{com relations}
\begin{split}
(u-v)[T_{ab}(u),T_{cd}(v)]=&\ (-1)^{|a||c|+|a||d|+|c||d|}\big(T_{cb}(v)T_{ad}(u)-T_{cb}(u)T_{ad}(v)\big)\\
=&\ (-1)^{|a||b|+|a||d|+|b||d|}\big(T_{ad}(u)T_{cb}(v)-T_{ad}(v)T_{cb}(u)\big).
\end{split}
\eeq

The super Yangian $\YglMN$ is a Hopf superalgebra with a coproduct and an opposite coproduct given by 
\beq\label{eq cop}
\begin{split}
&\Delta: T_{ab}(u)\mapsto \sum_{c=1}^{N} T_{cb}(u)\otimes T_{ac}(u),\qquad\  1\lle a,b\lle N,\\
&\widetilde\Delta: T_{ab}(u)\mapsto\sum_{c=1}^{N}(-1)^{(|a|+|c|)(|c|+|b|)} T_{ac}(u)\otimes T_{cb}(u),
\end{split}
\eeq
which have equivalent matrix forms
\beq\label{eq cop mat}
\begin{split}
&(\mathrm{id}\otimes\Delta)(T(u))=T^{(13)}(u)T^{(12)}(u),\\
&(\mathrm{id}\otimes\widetilde\Delta)(T(u))=T^{(12)}(u)T^{(13)}(u).
\end{split}
\eeq

For any complex number $z\in \C$, there is an automorphism
\beq\label{eq:shift}
\rho_z:\YglMN\to \YglMN,\qquad T_{ab}(u)\to T_{ab}(u-z),
\eeq
where $(u-z)^{-1}$ is expanded as a power series in $u^{-1}$. The \emph{evaluation homomorphism} is defined by the rule: 
\beq\label{eq:eval}
\epsilon:\YglMN\to \UglMN, \qquad T_{ba}^{\{s\}}\mapsto (-1)^{|a|}\delta_{1s}e_{ab},
\eeq
for $s\in\Z_{>0}$. The super Yangian $\YglMN$ contains $\UglMN$ as a Hopf subalgebra via the embedding given by $e_{ab}\mapsto (-1)^{|a|}T_{ba}^{\{1\}}$. By \eqref{com relations}, one has
\beq\label{zero mode com relations}
[T_{ab}^{\{1\}},T_{cd}(x)]=(-1)^{|a||c|+|a||d|+|c||d|}\big(\delta_{ad}T_{cb}(u)-\delta_{cb}T_{ad}(u)\big).
\eeq
The relation \eqref{zero mode com relations} implies that
\beq\label{eq com gll}
[E_{ab}\otimes 1+1\otimes e_{ab},T(x)]=0,
\eeq
for any $1\lle a,b\lle N$.

For any $\glMN$-module $M$ and $z\in\C$, denote by $M(z)$ the $\YglMN$-module obtained by pulling back of $M$ through the homomorphism $\pi(z):=\epsilon\circ\rho_z$. The module $M(z)$ is called an \emph{evaluation module at the evaluation point} $z$.

Let $\YMNplus$ be the left ideal of $\YglMN$ generated by the coefficients of the series $T_{ba}(u)$ for $1\lle a<b\lle N$. Let $\theta$ be the canonical projection $\theta:\YglMN\to \YglMN/\YMNplus$. For $A,B\in\YglMN$, we write $A\simeq B$ if $\theta(A)=\theta(B)$. In other words, $A\simeq B$ if $A-B\in \YMNplus$. The following lemma is straightforward by \eqref{com relations} and \eqref{eq cop}.

\begin{lem}\label{lem:plus-proj} 
For any $1\lle a,b\lle N$, we have
\begin{enumerate} 
\item the coefficients of the series $[T_{aa}(u),T_{bb}(v)]$ are in $\YMNplus$;
\item if $Z\in\YMNplus$, then the coefficients of $ZT_{aa}(u)$ are also in $\YMNplus$;
\item the coefficients of $\Delta(T_{aa}(u))-T_{aa}(u)\otimes T_{aa}(u)$ and $\widetilde\Delta(T_{aa}(u))-T_{aa}(u)\otimes T_{aa}(u)$ are in $\YMNplus\otimes \YglMN+\YglMN\otimes \YMNplus$.
\end{enumerate}
\end{lem}

We say that a vector $v$ in a $\YglMN$-module is a {\it singular $\ell$-weight vector} if $\YMNplus v=0$ and
\[
T_{aa}(u)v=\la_a(u)v, \qquad \la_a(u)\in 1+u^{-1}\C[u^{-1}].
\]
In this case, we call $\bla(u)=(\la_1(u),\dots,\la_N(u))$ the \emph{$\ell$-weight}  of $v$.

\subsection{Fusion procedure}\label{sec:fusion}
We recall the R-matrices defined by fusion procedure and their properties which will be used to define higher transfer matrices. Since the proofs of these statements are parallel to the even case, see \cite{MTV06}, {\cite[Propositions 1.6.2 \& 1.6.3]{Mol07}}, and references therein, we shall omit the details.

\begin{lem}\label{lem:anti-R-relation} We have
$$\displaystyle \mathop{\overrightarrow\prod}\limits_{1\lle i<j\lle k} R^{(ij)}(j-i)=\mathop{\overleftarrow\prod}\limits_{1\lle i<j\lle k} R^{(ij)}(j-i)=\mathbb H_{\{k\}}\prod_{j=1}^kj^{k-j+1},$$
$$\displaystyle \mathop{\overrightarrow\prod}\limits_{1\lle i<j\lle k} R^{(ij)}(i-j)=\mathop{\overleftarrow\prod}\limits_{1\lle i<j\lle k} R^{(ij)}(i-j)=\bA_{\{k\}}(-1)^k\prod_{j=1}^k(-j)^{k-j+1}.$$
\end{lem}

By the Yang-Baxter equation \eqref{eq:Yang-Baxter} and Lemma \ref{lem:anti-R-relation}, one has
\beq\label{eq-AAR}
\begin{split}
\bA_{\{k\}}^{(1\dots k)}\bA_{\{l\}}^{(k+1,\dots,k+l)}&\mathop{\overrightarrow\prod}\limits_{1\lle i\lle k}\mathop{\overleftarrow\prod}\limits_{1\lle j\lle l}	R^{(i,j+k)}(u+i-j-k+l)\\
=\Big(&\mathop{\overleftarrow\prod}\limits_{1\lle i\lle k}\mathop{\overrightarrow\prod}\limits_{1\lle j\lle l}	R^{(i,j+k)}(u+i-j-k+l)\Big)\bA_{\{k\}}^{(1\dots k)}\bA_{\{l\}}^{(k+1,\dots,k+l)}.
\end{split}
\eeq
Define $R^{\wedge k,\wedge l}(u)$ acting on $\scrV^{\wedge k}\otimes \scrV^{\wedge l}$ by
\[
R^{\wedge k,\wedge l}(u)=\mathop{\overleftarrow\prod}\limits_{1\lle i\lle k}\mathop{\overrightarrow\prod}\limits_{1\lle j\lle l}	R^{(i,j+k)}(u+i-j-k+l)\big|_{\scrV^{\wedge k}\otimes \scrV^{\wedge l}}\in \End(\scrV^{\wedge k})\otimes\End(\scrV^{\wedge l}).
\]
We have the following properties for these R-matrices,
\begin{align*}
\big(R^{\wedge k,\wedge l}(u-v)\big)^{(12)}&\big(R^{\wedge k,\wedge \ell}(u)\big)^{(13)}\big(R^{\wedge l,\wedge \ell}(v)\big)^{(23)}\\
= \ &\big(R^{\wedge l,\wedge \ell}(v)\big)^{(23)}\big(R^{\wedge k,\wedge \ell}(u)\big)^{(13)}\big(R^{\wedge k,\wedge l}(u-v)\big)^{(12)},
\end{align*}
and
\beq\label{eq:R-Q}
\big[R^{\wedge k,\wedge l}(u-v),Q^{\wedge k}\otimes Q^{\wedge l}]=0
\eeq
for any even matrix $Q\in\EndV$.

Let
\[
R_{\wedge k,\wedge 1}(u)=u+\sum_{a,b=1}^N\sum_{i=1}^k(E_{ab}^{(i)})\big|_{\sV^{\wedge k}}\otimes E_{ba}(-1)^{|b|}\in \End(\sV^{\wedge k})\otimes \EndV,
\]
\[
R_{\wedge 1,\wedge k}(u)=u+k-1+\sum_{a,b=1}^N\sum_{i=1}^kE_{ab}\otimes(E_{ba}^{(i)})\big|_{\sV^{\wedge k}} (-1)^{|b|}\in  \EndV\otimes\End(\sV^{\wedge k}).
\]
\begin{lem}\label{lem:antisym-sym-R}
We have
\[
R^{\wedge k,\wedge 1}(u)=R_{\wedge k,\wedge 1}(u)\prod\limits_{i=1}^{k-1}(u-i),\quad R^{\wedge 1,\wedge k}(u)=R_{\wedge 1,\wedge k}(u)\prod\limits_{i=0}^{k-2}(u+i).\qedd
\]	
\end{lem}

\begin{cor}\label{cor:fusion}
We have $R_{\wedge k,\wedge 1}(u)\big(R_{\wedge 1,\wedge k}(-u)\big)^{(21)}=(u+1)(k-u)$.
\end{cor}
\begin{proof}
This follows from the inversion relation of R-matrix, $R(u)R^{(21)}(-u)=1-u^2$.
\end{proof}

Consider the series $T^{(k,k+1)}(u)\cdots T^{(1,k+1)}(u-k+1)$ with coefficients in $\End(\scrV^{\otimes k})\otimes \YglMN$. By \eqref{eq RTT} and Lemma \ref{lem:anti-R-relation}, we have
\beq\label{eq:Att}
\begin{split}
\bA_{\{k\}}^{(1\dots k)}T^{(1,k+1)}(u-k+1)\cdots &\  T^{(k,k+1)}(u)\\
= &\ T^{(k,k+1)}(u)\cdots T^{(1,k+1)}(u-k+1)\bA_{\{k\}}^{(1\dots k)}.
\end{split}
\eeq
Hence the space $\scrV^{\wedge k}$ is invariant under all coefficients of the series $T^{(k,k+1)}(u)\cdots T^{(1,k+1)}(u-k+1)$. Denote $T^{\wedge k}(u)$ the restriction of this series to $\End(\scrV^{\wedge k})\otimes \YglMN$:
\beq\label{eq:T-wedge-K}
T^{\wedge k}(u)=T^{(k,k+1)}(u)\cdots T^{(1,k+1)}(u-k+1)\big|_{\scrV^{\wedge k}}^{(1\dots k)}.
\eeq
Note that, $T^{\wedge 1}(u)=T(u)$. Moreover, it follows from \eqref{eq RTT} and \eqref{eq-AAR} that
\beq\label{eq:high-RTT}
\begin{split}
\big(R^{\wedge k,\wedge l}(u-v)\big)^{(12)}&\big(T^{\wedge k}(u)\big)^{(13)}\big(T^{\wedge l}(v)\big)^{(23)}\\
= \ &\big(T^{\wedge l}(v)\big)^{(23)}\big(T^{\wedge k}(u)\big)^{(13)}\big(R^{\wedge k,\wedge l}(u-v)\big)^{(12)}.
\end{split}
\eeq
Clearly, by \eqref{eq cop mat}, we have
\beq\label{eq:copro-high}
\begin{split}
(\mathrm{id}\otimes\Delta)\big(T^{\wedge k}(u)\big)=\big(T^{\wedge k}(u)\big)^{(13)}\big(T^{\wedge k}(u)\big)^{(12)},\\
(\mathrm{id}\otimes\widetilde\Delta)\big(T^{\wedge k}(u)\big)=\big(T^{\wedge k}(u)\big)^{(12)}\big(T^{\wedge k}(u)\big)^{(13)}.
\end{split}
\eeq

\section{XXX spin chains}\label{sec:XXX}
\subsection{Higher transfer matrices}
For any even matrix $Q\in\EndV$, define the series
\beq\label{eq:def-transfer}
\mathscr T_{k,Q}(u)=(\str_{\scrV^{\wedge k}}\otimes \mathrm{id})\big(Q^{\wedge k}T^{\wedge k}(u)\big),\qquad k\in\Z_{>0},
\eeq
with coefficients in $\YglMN$. We call these series \emph{transfer matrices}. By convention, we also set $\mathscr T_{0,Q}(u)=1$. Note that transfer matrices are even. 
\begin{lem} Transfer matrices satisfy the following properties.
\begin{enumerate}
\item Transfer matrices commute, $[\mathscr T_{k,Q}(u),\mathscr T_{l,Q}(v)]=0$.
\item If $Q$ is the identity matrix, then coefficients of transfer matrices commute with the subalgebra $\UglMN$ in $\YglMN$.
\item If $Q$ is a diagonal matrix with pairwise distinct diagonal elements, then the subalgebra (the Bethe subalgebra) of $\YglMN$ generated by all coefficients of all transfer matrices contains $\mathrm U(\h)$, where $\h$ is the Cartan subalgebra of $\glMN$.\qed
\end{enumerate}
\end{lem}

One can also define another family of transfer matrices associated to symmetrizers
\[
\mathfrak T_{k,Q}(u)=(\str_{\scrV^{\otimes k}}\otimes \mathrm{id})\big(\mathbb H^{(1\dots k)}_{\{k\}}Q^{(1)}\cdots Q^{(k)}T^{(1,k+1)}(u)\cdots T^{(k,k+1)}(u-k+1)\big),\qquad k\in\Z_{>0}.
\]
Again, we set $\mathfrak T_{0,Q}(u)=1$.

Transfer matrices can be compactly combined into a generating series using quantum Berezinian as follows, see \cite{MR14} and cf. \cite{Tal06}.

We follow the convention of \cite{MR14}. Let $\mc A$ be a superalgebra. Consider the operators of the form
\beq\label{eq:matrix}
\mc K=
\sum_{a,b=1}^N(-1)^{|a||b|+|b|}E_{ab}\otimes K_{ab}\in \EndV\otimes \mc A,
\eeq
where $K_{ab}$ are elements of $\mc A$ of parity $|a|+|b|$. We say that $\mc K$ is a \emph{Manin matrix} if
\[
[K_{ab},K_{cd}]=(-1)^{|a||b|+|a||c|+|b||c|}[K_{cb},K_{ad}]
\]
for all $1\lle a,b,c,d\lle N$. 

If $\mc K$ is invertible and has the form
\[
\mc K^{-1}=\sum_{a,b=1}^N(-1)^{|a||b|+|b|}E_{ab}\otimes K'_{ab}\in \EndV\otimes \mc A,
\]
then we define the (quantum) \emph{Berezinian} of $\mc K$ by
\beq\label{eq:ber}
\mathrm{Ber}(\mc K)=\sum_{\sigma\in S_m}\mathrm{sign}(\sigma)K_{\sigma(1)1}\cdots K_{\sigma(m)m}\cdot\sum_{\tl\sigma\in S_n}\mathrm{sign}(\sigma)K'_{m+1,m+\tl\sigma(1)}\cdots K'_{m+n,m+\tl\sigma(n)}.
\eeq

Let $e^{-\pa_u}$ be the difference operator, $(e^{-\pa_u} f)(u)=f(u-1)$ for any function $f$ in $u$. Let $\mc A$ be the superalgebra $\YglMN[[u^{-1},\pa_u]]$, where $\pa_u$ is even. Here $u$ and $\pa_u$ satisfy the relations
$$
\pa_u\cdot u^{-s}=u^{-s}\pa_u-su^{-s-1},\quad s\in\Z_{>0}.
$$
Consider the operator $Z^Q(x,\pa_u)$,
\[
Z_Q(u,\pa_u)=T^\dag(u)Q^\dag e^{-\pa_u}\in \EndV\otimes\YglMN[[u^{-1},\pa_u]].
\]
It follows from \eqref{eq RTT} or \eqref{com relations} that $Z_Q(u,\pa_u)$ is a Manin matrix, see e.g. \cite[Remark 2.12]{MR14} and cf. \cite[Proposition 4]{CF08}. Note that our generating series $T_{ij}(u)$ corresponds to $z_{ji}(u)$ in \cite{MR14}.

Define the rational difference operator $\mc D^Q(u,\pa_u)$,
\beq\label{eq:Ber-XXX}
\mc D_Q(u,\pa_u)=\mathrm{Ber}(1-Z_Q(u,\pa_u)).
\eeq
Applying the supertransposition to all copies of $\EndV$ and using cyclic property of supertrace, see \eqref{eq transpose cyclic}, it follows from \cite[Theorem 2.13]{MR14} that 
\beq\label{eq diff trans}
\mc D_Q(u,\pa_u)=\sum_{k=0}^\infty (-1)^k \mathscr T_{k,Q}(u)e^{-k\pa_u},\quad \big(\mc D_Q(u,\pa_u)\big)^{-1}=\sum_{k=0}^\infty \mathfrak T_{k,Q}(u)e^{-k\pa_u}.
\eeq

\subsection{Universal off-shell Bethe vectors}\label{sec:bv}
In this section, we recall the supertrace formula of Bethe vectors and its properties from \cite{MTV06,BR08,PRS17}.

Let $\bm\xi=(\xi^1,\dots,\xi^{N-1})$ be a sequence of nonnegative integers. Set $\xi^{< a}=\xi^1+\cdots+\xi^{a-1}$, $1< a \lle N$. In particular, we set $|\bm\xi|=\xi^{<N}$. Consider a series in $|\bm\xi|$ variables $$\bm t=(t_1^1,\dots,t^{1}_{\xi^1},\dots,t_1^{N-1},\dots,t^{N-1}_{\xi^{N-1}})$$ with coefficients in $\YglMN$,
\beq\label{EQ:BAEU}
\begin{split}
\widehat{\bB}_{\bm\xi}(\bm t)=\ &(\str\otimes \mathrm{id})\Big(T^{(1,|\bm\xi|+1)}(t_1^1)\cdots T^{(|\bm\xi|,|\bm\xi|+1)}(t_{\xi^{N-1}}^{N-1})\\
&\times \mathop{\overrightarrow\prod}\limits_{(a,i)<(b,j)}R^{(\xi^{<b}+j,\xi^{<a}+i)}(t_j^b-t_i^a)E_{21}^{\otimes \xi^1}\otimes\cdots\otimes E_{N,N-1}^{\otimes \xi^{N-1}}\otimes 1\Big),
\end{split}
\eeq
where the supertrace is taken over all factors and the pairs are ordered lexicographically, namely $(a,i)<(b,j)$ if $a<b$, or $a=b$ and $i<j$. Moreover, the product is taken over the set $\{(c,k)~|~1\lle c < N,\ 1\lle k\lle \xi^c\}$. Note that Bethe vectors are obtained by applying $\widehat{\bB}_{\bm\xi}(\bm t)$ to pseudovaccum vectors (singular $\ell$-weight vectors). Therefore, we call $\widehat{\bB}_{\bm\xi}(\bm t)$ (and its renormalizations) a \emph{universal off-shell Bethe vector}.

This supertrace formula is slightly different from the one in \cite{BR08}. But it turns out that they only differ by a scalar function in $\bm t$, see \cite[Propositions 3.2 \& 3.3]{PRS17}. 

It is clear from the Yang-Baxter equation and the equality $$R(u-v)E_{ab}\otimes E_{ab}=(u-v+(-1)^{|a|})E_{ab}\otimes E_{ab}$$ that $\widehat{\bB}_{\bm\xi}(\bm t)$ is divisible by 
\beq\label{eq-factors}
\prod_{a=1}^{N-1} \prod_{1\lle i<j\lle \xi^a}(t_j^a-t_i^a+(-1)^{|a+1|})
\eeq
in $\YglMN[t_1^1,\dots,t^{N-1}_{\xi^{N-1}}][[(t_1^1)^{-1},\dots,(t^{N-1}_{\xi^{N-1}})^{-1}]]$.

Recall the canonical projection $\theta:\YglMN\to \YglMN/\YMNplus$.
\begin{lem}\label{lem:div-far}
The series $\theta(\widehat{\bB}_{\bm\xi}(\bm t))$ is divisible by
\[
\prod_{a=1}^{N-2}\prod_{b=a+2}^{N-1}\prod_{i=1}^{\xi^a}\prod_{j=1}^{\xi^b}(t_j^b-t_i^a)
\]
in $(\YglMN/\YMNplus)[t_1^1,\dots,t^{N-1}_{\xi^{N-1}}][[(t_1^1)^{-1},\dots,(t^{N-1}_{\xi^{N-1}})^{-1}]]$.\qed
\end{lem}
The lemma will be proved in Section \ref{sec:bv-recur} after Proposition \ref{prop:BV-rec-XXX} where we recall the recursion for the Bethe vector.

Set
\beq\label{eq:bae-sym}
\begin{split}
&\bB_{\bm\xi}(\bm t)=\widehat{\bB}_{\bm\xi}(\bm t)\prod_{a=1}^{N-1} \prod_{1\lle i<j\lle \xi^a}\frac{1}{t_j^a-t_i^a+(-1)^{|a+1 |}} \prod_{1\lle a< b<N}\prod_{i=1}^{\xi^a}\prod_{j=1}^{\xi^b}\frac{1}{t_j^b-t_i^a},\\
&\overline\bB_{\bm\xi}(\bm t)=\bB_{\bm\xi}(\bm t)\prod_{1\lle i<j\lle \xi^m}\frac{1}{t_j^m-t_i^m-(-1)^{|m+1 |}},
\end{split}
\eeq
see \cite[Equation (3.1) \& Proposition 3.3]{PRS17}. Note that $\overline\bB_{\bm\xi}(\bm t)$ corresponds to the Bethe vector used in \cite{PRS17} which is  symmetric in variables $t_i^a$ with the same superscript $a$ for all $1\lle a< N$, see \cite[Proposition 3.2]{PRS17}. Examples of $\overline\bB_{\bm\xi}(\bm t)$ for small $N$ can be found in \cite[Section 3.1]{PRS17}.  

Here we shall mainly use $\bB_{\bm\xi}(\bm t)$ with $\bm t$ satisfying
\beq\label{eq:assm}
t_j^m-t_i^m-(-1)^{|m+1 |}\ne 0,\quad 1\lle i<j\lle \xi^m,
\eeq
due to the equality $T_{m+1,m}(u)T_{m+1,m}(u-1)=0$. Note that equation \eqref{eq:assm} always holds after reordering $t_i^m$ e.g. in increasing order with respect to the real parts. Moreover, $\bB_{\bm \xi}(\bm t)$ is symmetric in variables $t_i^a$ with the same superscript $a$ for all $1\lle a< N$ except when $a=m$ (the corresponding simple root is odd). 


In general, $\bB_{\bm\xi}(\bm t)$ is a sum of the products
\beq\label{eq:summand}
\begin{split}
T_{a_1,b_1}(t_1^1)\cdots T_{a_{|\bm\xi|},b_{|\bm\xi|}}(t_{\xi^{N-1}}^{N-1})p(\bm t)&\prod_{1\lle a< b<N}\prod_{i=1}^{\xi^a}\prod_{j=1}^{\xi^b}\frac{1}{t_j^b-t_i^a}
\end{split}
\eeq
with various $a_1,\dots, a_{|\bm\xi|}$, $b_1,\dots, b_{|\bm\xi|}$ from $\{1,\dots,N\}$ and polynomials $p(\bm t)$.

\subsection{Main technical result}
We use the same notation as in Section \ref{sec:bv}. Following e.g. \cite{HMVY19}, introduce a sequence of polynomials $\bm y=(y_1,\dots,y_{N-1})$ associated to $\bm t$ and $\bm\xi$,
\beq\label{eq:y-t}
y_a(u)=\prod_{i=1}^{\xi^a}(u-t_i^a).
\eeq
By convention $y_0=y_N=1$. We also set $\ka_a=1$ for $1\lle a\lle m$ and $\ka_b=-1$ for $m<b\lle N$. 

From now on, we assume that $Q=\sum_{a=1}^NQ_aE_{aa}$ is diagonal. Define the series
\beq\label{eq:I}
\begin{split}
\mathscr I_{\bm\xi,Q}^{a,i}(\bm t)=\ka_aQ_aT_{aa}(t_i^a)&y_{a-1}(t_i^a+\ka_a)y_{a}(t_i^a-\ka_a)y_{a+1}(t_i^a)\\
+&\ \ka_{a+1}Q_{a+1}T_{a+1,a+1}(t_i^a)y_{a-1}(t_i^a)y_{a}(t_i^a+\ka_{a+1})y_{a+1}(t_i^a-\ka_{a+1})
\end{split}
\eeq
for $1\lle a< N$, $1\lle i\lle \xi^a$.

Given the data: integers $a_1,\dots,a_{|\bm\xi|+k-1}$, $b_1,\dots,b_{|\bm \xi|+k-1}$, $c\in \{1,\dots,N\}$ and $i\in \{1,\dots,\xi^c\}$, a sequence $s_1,\dots,s_{|\bm\xi|+k-1}$ which is a permutation of the sequence $u,\dots,u-k+1$, $t_1,\dots,\widehat{t_i^c},\dots,t_{\xi^{N-1}}^{N-1}$, where the hat means that the corresponding variable $t_i^c$ is skipped, and a polynomial $p(u;\bm t)$, consider
\beq\label{eq:products-I}
\begin{split}
T_{a_1,b_1}(s_1)\cdots &T_{a_{|\bm\xi|+k-1},b_{|\bm\xi|+k-1}}(s_{|\bm\xi|+k-1})\mathscr I_{\bm \xi,Q}^{c,i}(\bm t)p(u;\bm t)\\
\times &\ \prod_{a=1}^{N-1}\prod_{j=1}^{\xi^a}\Big(\frac{1}{u-t_j^a}\prod_{1\lle j< l\lle \xi^a}\frac{1}{t_j^a-t_l^a}\Big)\prod_{a=1}^{N-2}\prod_{j=1}^{\xi^a}\prod_{l=1}^{\xi^{a+1}}\frac{1}{(t_{l}^{a+1}-t_j^a)^2}.
\end{split}
\eeq
Here the factors $(u-t_j^a)^{-1}$ are considered as power series in $u^{-1}$. Denote by $I_{\bm\xi,\bm t,k,Q}$ the $\C$-span of all products \eqref{eq:products-I} with all possible data. We also denote by $I_{\bm\xi,\bm t,Q}$ the sum of $I_{\bm\xi,\bm t,k,Q}$ for $k\in\Z_{>0}$.

For $1\lle a\lle N$, define the series
\beq\label{eq:chi}
\mathscr X_{\bm\xi,Q}^a(u;\bm t)=Q_aT_{aa}(u)\frac{y_{a-1}(u+\ka_a)y_a(u-\ka_a)}{y_{a-1}(u)y_a(u)},
\eeq
which are regarded as power series in $u^{-1}$ with coefficients in $\YglMN[t_1^1,\dots,t_{\xi^{N-1}}^{N-1}]$. 

\begin{thm}\label{thm:main-tech}
Let $Q$ be a diagonal matrix. Then we have
\[
\mathscr T_{k,Q}(u)\bB_{\bm\xi}(\bm t)\simeq\bB_{\bm\xi}(\bm t)\sum_{\bm a}\prod_{r=1}^k\ka_{a_r}\mathscr X_{\bm\xi,Q}^{a_r}(u-r+1;\bm t)+\mathscr U_{\bm \xi,k,Q}(u;\bm t),
\]
where the sum is taken over all $k$-tuples $\bm a=(1\lle a_1<\cdots<a_b<m+1\lle a_{b+1}\lle \cdots\lle a_k\lle N)$ for various $0\lle b\lle k$ and $\mathscr U_{\bm \xi,k,Q}(u;\bm t)$ is in $I_{\bm\xi,\bm t,k,Q}$.
\end{thm}

The theorem is proved in Section \ref{app:A}. 

Note that, due to Lemma \ref{lem:plus-proj}, the order of $\mathscr X_{\bm\xi,Q}^{a_r}(u-r+1;\bm t)$ in Theorem \ref{thm:main-tech} is irrelevant.

\begin{cor}\label{cor:main-tech}
Let $Q$ be a diagonal matrix. Then we have \[
\mc D_Q(u,\pa_u)\bB_{\bm\xi}(\bm t)\simeq\bB_{\bm\xi}(\bm t)\mathop{\overrightarrow\prod}\limits_{1\lle a\lle N}\big(1-\mathscr X_{\bm\xi,Q}^a(u;\bm t)e^{-\pa_u}\big)^{\ka_a}+\mathscr U_{\bm \xi,Q}(u;\bm t),
\]	
where $\mathscr U_{\bm \xi,Q}(u;\bm t)$ belongs to $I_{\bm\xi,\bm t,Q}$ and $\mc D_Q(u,\pa_u)$is defined in \eqref{eq:Ber-XXX}.
\end{cor}
\begin{proof}
This follows from direct computations using \eqref{eq diff trans} and Theorem \ref{thm:main-tech}.	
\end{proof}

\subsection{Main results for XXX spin chains}
In this section, we shall obtain eigenvectors and eigenvalues of transfer matrices when the underlying Hilbert spaces are tensor products of evaluation modules of the super Yangian $\YglMN$, proving \cite[Conjecture 5.15]{LM21b}. For more general highest weight representations of $\YglMN$, \cite[Conjecture 5.15]{LM21b} is proved similarly for generic situation.

Let $\ell$ be a positive integer. Note that $\ell$ here has nothing to do with $\ell$ in $\ell$-weight. Let $M_1,\dots,M_{\ell}$ be $\glMN$-modules, $\bm z=(z_1,\dots,z_{\ell})$ a sequence of complex numbers. Consider the tensor product of evaluation $\YglMN$-modules,
$$
M(\bm z):= M_1(z_1)\otimes \cdots \otimes M_\ell(z_\ell).
$$
Then, by \eqref{eq:shift} and \eqref{eq:eval}, the operator
\beq\label{eq:transfer-XXX}
\mathscr T_{k,Q}^{M}(u;\bm z)=\mathscr T_{k,Q}(u)\big|_{M(\bm z)}
\eeq
is a rational function in $u,\bm z$ with denominators $\prod_{i=1}^\ell\prod_{j=0}^{k-1}(u-j-z_i)$. Note that
\[
\mathscr T_{k,Q}^{M}(u;\bm z)=\str_{\scrV^{\wedge k}}Q^{\wedge k}+O(u^{-1}),\quad u\to\infty.
\]
We call the operators $\mathscr T_{k,Q}^{M}(u;\bm z)$, $k\in\Z_{>0}$, the \emph{transfer matrices of the XXX spin chain} on $M(\bm z)$ associated with $\glMN$.

We are interested in the case when $M_1,\dots, M_\ell$ are highest weight $\glMN$-modules with highest weights $\La_1,\dots,\La_\ell$, where $\La_i=(\La_i^1,\dots,\La_i^N)$, and highest weight vectors $v_1,\dots,v_\ell$, respectively.
By convention, we set $\bLa=(\La_1,\dots,\La_\ell)$.

In this case, the vector $ v^+=v_1\otimes\cdots\otimes v_{\ell}$ is a singular $\ell$-weight vector of $\YglMN$ in $M(\bm z)$ whose $\ell$-weight is given as follows,
\beq\label{eq:high-ell-wt}
T_{aa}(u)v^+=v^+\ \prod_{i=1}^\ell\frac{u-z_i+\ka_a\La_i^a}{u-z_i},\quad 1\lle a\lle N.
\eeq

Recall $\bB_{\bm\xi}(\bm t)$ and the notations from Section \ref{sec:bv}. Apply $\bB_{\bm\xi}(\bm t)$ to $v^+$ and renormalize it so that the function
\beq\label{eq:off-shell-bv}
\bB_{\bm \xi}^{v^+}(\bm t;\bm z)=\bB_{\bm \xi}(\bm t)v^+\ \prod_{a=1}^{N-1}\prod_{i=1}^{\xi^a}\prod_{j=1}^{\ell}(t_i^a-z_j)\prod_{a=1}^{N-2}\prod_{i=1}^{\xi^a}\prod_{j=1}^{\xi^{a+1}}(t_j^{a+1}-t_i^a)
\eeq
is a polynomial in $\bm t$, $\bm z$, see Lemma \ref{lem:div-far} and \eqref{eq:summand}. We call $\bB_{\bm \xi}^{v^+}(\bm t;\bm z)$ the \emph{off-shell Bethe vector} for the XXX spin chain on $M(\bm z)$ associated with $\glMN$.

Recall that $Q=\sum_{a=1}^NQ_aE_{aa}$ is diagonal and let $\bm y$ be the sequence of polynomials associated to $\bm t$ and $\bm \xi$. Consider the system of algebraic equations
\beq\label{eq:bae}
\begin{split}
&-\ka_aQ_a\prod_{j=1}^\ell(t_i^a-z_j+\ka_a\La_j^a)y_{a-1}(t_i^a+\ka_a)y_{a}(t_i^a-\ka_a)y_{a+1}(t_i^a)\\
=&\ \ka_{a+1}Q_{a+1}\prod_{j=1}^\ell(t_i^a-z_j+\ka_{a+1}\La_j^{a+1})y_{a-1}(t_i^a)y_{a}(t_i^a+\ka_{a+1})y_{a+1}(t_i^a-\ka_{a+1}),
\end{split}
\eeq
where $1\lle a<N$, $1\lle i\lle \xi^a$, $y_0=y_N=1$. We call \eqref{eq:bae} the \emph{Bethe ansatz equation}. We say that a solution $\tl{\bm t}=(\tl{t}_1^1,\dots,\tl{t}_{\xi^{N-1}}^{N-1})$ of the Bethe ansatz equation \eqref{eq:bae} is \emph{off-diagonal} if $\tl t_i^{a}\ne \tl t_j^{a}$ for any $1\lle a<N$, $1\lle i<j\lle \xi^a$ and $\tl t_i^{a}\ne \tl t_j^{a+1}$ for any $1\lle a<N-1$, $1\lle i\lle \xi^a$, $1\lle j\lle \xi^{a+1}$ (aslo \eqref{eq:assm}).

When $\tl{\bm t}$ is an off-diagonal solution of the Bethe ansatz equation \eqref{eq:bae}, we say that the vector $\bB_{\bm \xi}^{v^+}(\tl{\bm t};\bm z)$ is an \emph{on-shell Bethe vector}.

For $1\lle a\lle N$, define 
\beq\label{eq:chi-special}
\mathscr X_{\bm\xi,Q}^a(u;\bm t;\bm z;\bLa)=Q_a\frac{y_{a-1}(u+\ka_a)y_a(u-\ka_a)}{y_{a-1}(u)y_a(u)}\prod_{i=1}^\ell\frac{u-z_i+\ka_a\La_i^a}{u-z_i},
\eeq
where $\bm y=(y_1,\dots,y_{N-1})$ is the sequence of polynomials associated to $\bm t$ and $\bm \xi$.

\begin{thm}\label{thm:main-tech-xxx}
Let $Q$ be a diagonal matrix. If $M_1,\dots, M_\ell$ are highest weight $\glMN$-modules with highest weights $\La_1,\dots,\La_\ell$ and $\tl{\bm t}$ is an off-diagonal solution of the Bethe ansatz equation \eqref{eq:bae}, then we have
\[
\mathscr T_{k,Q}(u)\bB_{\bm\xi}^{v^+}(\tl{\bm t};\bm z)=\bB_{\bm\xi}^{v^+}(\tl{\bm t};\bm z)\sum_{\bm a}\prod_{r=1}^k\ka_{a_r}\mathscr X_{\bm\xi,Q}^{a_r}(u-r+1;\tl{\bm t};\bm z;\bLa),
\]
where the sum is taken over all $k$-tuples $\bm a=(1\lle a_1<\cdots<a_b<m+1\lle a_{b+1}\lle \cdots\lle a_k\lle N)$ for various $0\lle b\lle k$. 
\end{thm}

The proof of the theorem is given in Section \ref{app:A}. The statement was shown in \cite[Theorem 5.4]{MTV06} for the general even case and conjectured in \cite[Conjecture 5.15]{LM21b}. The case of $m=n=1$ was previously shown in \cite[Theorem 6.1]{LM21a}. It can be thought as the supersymmetric version of \cite[Theorem 5.11]{FH:2015} and \cite[Theorem 7.5]{FJMM17} for type A. When $k=1$, the statement was obtained in \cite{BR08}.

\begin{cor}\label{cor:main-tech-xxx}
Let $Q$ be a diagonal matrix. If $M_1,\dots, M_\ell$ are highest weight $\glMN$-modules with highest weights $\La_1,\dots,\La_\ell$ and $\tl{\bm t}$ is an off-diagonal solution of the Bethe ansatz equation \eqref{eq:bae}, then we have
\[
\mc D_Q(u,\pa_u)\bB_{\bm\xi}^{v^+}(\tl{\bm t};\bm z)=\bB_{\bm\xi}^{v^+}(\tl{\bm t};\bm z)\mathop{\overrightarrow\prod}\limits_{1\lle a\lle N}\big(1-\mathscr X_{\bm\xi,Q}^a(u;\tl{\bm t};\bm z;\bLa)e^{-\pa_u}\big)^{\ka_a},
\]
where $\mc D_Q(u,\pa_u)$is defined in \eqref{eq:Ber-XXX}.	
\end{cor}

Note that the rational difference operator 
\beq\label{eq:rational-diff-oper}
\mc D_Q(u,\pa_u;\tl{\bm t};\bm z;\bLa):=\mathop{\overrightarrow\prod}\limits_{1\lle a\lle N}\big(1-\mathscr X_{\bm\xi,Q}^a(u;\tl{\bm t};\bm z;\bLa)e^{-\pa_u}\big)^{\ka_a}
\eeq
was introduced in \cite[Equation (5.6)]{HLM19}, cf. \cite[Equations (2.13)]{Tsu97}.

\begin{prop}\label{prop:singular}
Let $Q$ be the identity matrix. If $M_1,\dots, M_\ell$ are highest weight $\glMN$-modules with highest weights $\La_1,\dots,\La_\ell$ and $\tl{\bm t}$ is an off-diagonal solution of the Bethe ansatz equation \eqref{eq:bae}, then the on-shell Bethe vector $\bB_{\bm\xi}^{v^+}(\tl{\bm t};\bm z)$ is a $\glMN$-singular vector in $M_1\otimes\cdots\otimes M_\ell$ with weight
$$
\Big(\sum_{i=1}^\ell\La_i^1-\xi^1,\sum_{i=1}^\ell\La_i^1+\xi^1-\xi^2,\dots,\sum_{i=1}^\ell\La_i^{N}+\xi^{N-1}\Big).
$$
\end{prop}
The proof of the proposition is given in Section \ref{app:A}.

\section{Proof of main results}\label{sec:more}
We start with preparing a few statements which will be used in the proof.

\subsection{Recursion for the Bethe vectors}\label{sec:bv-recur}

Since we shall use nested algebraic Bethe ansatz, many notations will be used for both $\gl_{m|n}$ and $\gl_{m-1|n}$. To simplify the notation, we use $\cN$ and $\cN-1$ to distinguish notations for $\gl_{m|n}$ and $\gl_{m-1|n}$, respectively. We also use $\langle \cN\rangle$ and $\langle \cN-1\rangle$.

Set $\mathscr W=\C^{\cN-1}=\C^{m-1|n}$. Let $\bfw_1,\dots,\bfw_{N-1}$ be the standard basis of $\sW$ and $\bfv_1,\dots,\bfv_{N}$ of $\scrV=\C^{\cN}=\C^{m|n}$. Identify $\sW$ with the subspace of $\sV$ via the embedding $\bfw_a\mapsto \bfv_{a+1}$, $1\lle a< N$.

Let $P^{\langle\cN-1\rangle}\in \End(\sW^{\otimes 2})$ be the graded flip operator and $R^{\langle\cN-1\rangle}(u)=u+P^{\langle\cN-1\rangle}$ be the rational R-matrix used to define the super Yangian $\mathrm{Y}(\gl_{\cN-1})$. The R-matrix $R(u)$ preserves the subspace $\sW^{\otimes 2}\subset \sV^{\otimes 2}$ and the restriction of $R(u)$ on $\sW^{\otimes 2}$ coincides with $R^{\langle\cN-1\rangle}(u)$. Recall that $\sW(x)$ is the evaluation $\mathrm{Y}(\gl_{\cN-1})$-module with the corresponding homomorphism $\pi(x):\mathrm{Y}(\gl_{\cN-1})\to \End(\sW)$,
\beq\label{eq:def:eva-2}
\pi(x): T^{\langle\cN-1\rangle}(u)\mapsto (u-x)^{-1}R^{\langle\cN-1\rangle}(u-x).
\eeq

Define the embedding $\psi:\mathrm{Y}(\gl_{\cN-1})\hookrightarrow\mathrm{Y}(\gl_{\cN})$ by the rule $\psi(T_{ab}^{\langle\cN-1\rangle}(u))=T_{a+1,b+1}(u)$, $1\lle a,b\lle N-1$. Note that $\psi(\mathrm{Y}_+(\gl_{\cN-1}))\subset \mathrm{Y}_+(\gl_{\cN})$. 

Define a map $\psi(x_1,\dots,x_r):\mathrm{Y}(\gl_{\cN-1})\to \mathrm Y(\gl_{\cN})\otimes \End(\sW^{\otimes r})$ by
\beq\label{eq:def-psi}
\psi(x_1,\dots,x_r)=(\psi\otimes\pi(x_r)\otimes\cdots\otimes\pi(x_1))\circ(\widetilde\Delta^{\langle\cN-1\rangle})^{(r)},
\eeq
where $(\widetilde\Delta^{\langle\cN-1\rangle})^{(r)}:\mathrm Y(\gl_{\cN})\to \mathrm Y(\gl_{\cN})^{\otimes(r+1)}$ is the multiple opposite coproduct. Note that here we use the opposite coproduct $\widetilde\Delta$ which is consistent with that in \cite{BR08}.

Define a map $\widetilde \psi:\mathrm{Y}(\gl_{\cN-1})\to \mathrm Y(\gl_{\cN})\otimes\sW^{\otimes r}$ by
\[
\widetilde \psi(x_1,\dots,x_r)= \psi(x_1,\dots,x_r)(1\otimes\bfw_1^{\otimes r}).
\]
The following lemmas are straightforward.
\begin{lem}\label{lem11.1}
	We have $\widetilde \psi(x_1,\dots,x_r)(\mathrm{Y}_+(\gl_{\cN-1}))\subset \mathrm Y_+(\gl_{\cN})\otimes\sW^{\otimes r} $.
\end{lem}

Similarly, define the embedding $\phi:\mathrm{Y}(\gl_{\cN-2})\hookrightarrow\mathrm{Y}(\gl_{\cN-1})$ by the rule $$\phi(T_{ab}^{\langle\cN-2\rangle}(u))=T_{a+1,b+1}^{\langle\cN-1\rangle}(u),\quad 1\lle a,b\lle N-2.$$ Recall the canonical projection $\theta:\mathrm Y(\gl_{\cN})\twoheadrightarrow \mathrm Y(\gl_{\cN})/\mathrm Y_+(\gl_{\cN})$.

\begin{lem}\label{lem11.2}
	We have $(\theta\otimes\id^{\otimes r})\widetilde \psi(x_1,\dots,x_r)\circ\phi = (\theta\circ \psi\circ \phi)\otimes\bfw_1^{\otimes r}$.
\end{lem}

Set $\bar{\bm\xi}=(\xi^2,\dots,\xi^{N-1})$ and $\bar{\bm t}=(t_1^2,\dots,t_{\xi^2}^2;\dots;t_1^{N-1},\dots,t_{\xi^{N-1}}^{N-1})$.
\begin{prop}[{\cite[Eq. (5.1)]{BR08}}]\label{prop:BV-rec-XXX}
We have	
\[
\bB_{\bm\xi}(\bm t)=B^{(1)}(t_1^1)\cdots B^{(\xi^1)}(t^1_{\xi^1})\widetilde \psi(t_1^1,\dots,t^1_{\xi^1})\big(\bB^{\langle\cN-1\rangle}_{\bar{\bm\xi}}(\bar{\bm t})\big)
\]
where $B(u)=(T_{12}(u),\dots,T_{1N}(u))=\sum_{a=1}^{N-1} E_{1,a+1}\otimes T_{1,a+1}(u)$ and its coefficients are treated as elements in $\mathrm{Hom}(\sW,\C)\otimes \YglMN$. 
\end{prop}

\begin{proof}[Proof of Lemma \ref{lem:div-far}]
By Lemma \ref{lem11.2}, Proposition \ref{prop:BV-rec-XXX}, \eqref{eq:def:eva-2} and \eqref{eq:def-psi}, the denominator of $\theta(\bB_{\bm\xi}(\bm t))$ is at most 
\[
\prod_{a=1}^{N-2}\prod_{i=1}^{\xi^a}\prod_{j=1}^{\xi^{a+1}}(t_j^{a+1}-t_i^a).
\]
Then the lemma follows from \eqref{eq:bae-sym}.
\end{proof}

\subsection{Correspondence between $\YglMN$ and $\mathrm Y(\gl_{n|m})$}
\label{sec:correspondence}
For $1\lle a\lle N$, set $a'=N+1-a$. In the following, we use $a'$ and $b'$ for the indices corresponding to the super Yangian $\mathrm{Y}(\gl_{n|m})$. Moreover, their parities should be the parities inherited from $\mathrm{Y}(\gl_{n|m})$. We have the isomorphism 
$$
\varpi:\mathrm{Y}(\gl_{m|n})\to \mathrm{Y}(\gl_{n|m}),
\quad T_{ab}(u)\to \widetilde{T}_{b'a'}(u)(-1)^{|a'||b'|+|b'|}.
$$
Here and below, we shall use tilde to indicate the notations corresponding to $\mathrm{Y}(\gl_{n|m})$. Note that $\varpi$ maps $\mathrm{Y}_+(\gl_{m|n})$ to $\mathrm{Y}_+(\gl_{n|m})$.

We now describe the image of the rational difference operator \eqref{eq diff trans} under the isomorphism $\varpi$. Recall the transfer matrices associated to symmetrizers, 
\[
\mathfrak T_{k,Q}(u)=(\str_{\scrV^{\otimes k}}\otimes \mathrm{id})\big(\mathbb H^{(1\dots k)}_{\{k\}}Q^{(1)}\cdots Q^{(k)}T^{(1,k+1)}(u)\cdots T^{(k,k+1)}(u-k+1)\big),
\]
for $k\in\Z_{>0}$.

Set $\widetilde\scrV=\C^{n|m}$. We can identify $\scrV$ and $\widetilde{\scrV}$ by identifying $\bfv_a$ with $\widetilde{\bfv}_{a'}$, $1\lle a\lle N$. Note that the parities are changed under this identification and the operator $E_{ab}$ is identified with the operator $E_{a'b'}$. Moreover,
\[
(\str_{\scrV}\otimes \mathrm{id})(E_{ab})=(-1)^{|a|}\delta_{ab}=-(-1)^{|a'|}\delta_{ab}=-(\str_{\widetilde\scrV}\otimes \mathrm{id})(E_{a'b'}).
\]
\begin{eg}
We have the following identification 
$$
Q=\sum_{a=1}^N Q_{a}E_{aa}\in\EndV\longrightarrow \mathsf Q=\sum_{a=1}^N Q_{a}E_{a'a'}\in \End(\widetilde{\scrV}).
$$
Due to $(-1)^{|a|}=-(-1)^{|a'|}$ (i.e. $\ka_a=-\tilde\ka_{a'}$), the R-matrix $R(u)\in \End(\scrV^{\otimes 2})$ is identified with $-\widetilde{R}(-u)\in \End(\widetilde\scrV^{\otimes 2})$.
In particular, by Lemma \ref{lem:anti-R-relation}, the action of $\mathbb H_{\{k\}}$ on $\scrV^{\otimes k}$ is the same as that of $\widetilde{\mathbb A}_{\{k\}}$ on ${\widetilde{\scrV}}^{\otimes k}$.
\end{eg}

Define the matrix 
\[
\widetilde{\mathsf T}(u):=\big(\widetilde{T}(u)\big)^\dag=\sum_{a',b'=1}^N E_{a'b'}\otimes \widetilde{T}_{b'a'}(u)(-1)^{|a'||b'|+|b'|}\in \End(\tilde\scrV)\otimes \mathrm{Y}(\gl_{0|n})[[u^{-1}]].
\]
Observe that 
\beq\label{eq:pi-T}
\varpi(T(u))=\widetilde{\mathsf T}(u).
\eeq
\begin{lem}
We have	$\varpi(\mathfrak T_{k,Q}(u))=(-1)^k\widetilde{\mathscr T}_{k,\mathsf Q}(u)$ and $\varpi(\mathscr T_{k,Q}(u))=(-1)^k\widetilde{\mathfrak T}_{k,\mathsf Q}(u)$.
\end{lem}
\begin{proof}
By $\mathsf Q^\dag=\mathsf Q$, $\big(\widetilde{\mathbb A}_{\{k\}}\big)^\dag= \widetilde{\mathbb A}_{\{k\}}$, and the cyclicity of supertrace, we have
\begin{align*}
\varpi(\mathfrak T_{k,Q}(u))=&\ (-1)^k(\str_{{\widetilde\scrV}^{\otimes k}}\otimes \mathrm{id})\big(\widetilde{\mathbb A}^{(1\dots k)}_{\{k\}}\mathsf Q^{(1)}\cdots \mathsf Q^{(k)}\widetilde{\mathsf T}^{(1,k+1)}(u)\cdots \widetilde{\mathsf T}^{(k,k+1)}(u-k+1)\big)\\
=&\ (-1)^k(\str_{{\widetilde\scrV}^{\otimes k}}\otimes \mathrm{id})\big(\widetilde{T}^{(1,k+1)}(u)\cdots \widetilde{T}^{(k,k+1)}(u-k+1)\mathsf Q^{(1)}\cdots \mathsf Q^{(k)}\widetilde{\mathbb A}^{(1\dots k)}_{\{k\}}\big)\\
=&\ (-1)^k(\str_{{\widetilde\scrV}^{\otimes k}}\otimes \mathrm{id})\big(\widetilde{T}^{(1,k+1)}(u)\cdots \widetilde{T}^{(k,k+1)}(u-k+1)\widetilde{\mathbb A}^{(1\dots k)}_{\{k\}}\mathsf Q^{(1)}\cdots \mathsf Q^{(k)}\big)\\
=&\ (-1)^k(\str_{{\widetilde\scrV}^{\otimes k}}\otimes \mathrm{id})\big(\mathsf Q^{(1)}\cdots \mathsf Q^{(k)}\widetilde{T}^{(1,k+1)}(u)\cdots \widetilde{T}^{(k,k+1)}(u-k+1)\widetilde{\mathbb A}^{(1\dots k)}_{\{k\}}\big)\\
=&\ (-1)^k\widetilde{\mathscr T}_{k,\mathsf Q}(u).
\end{align*}
The other one is similar.
\end{proof}
In particular, by \eqref{eq diff trans}, we obtain 
\begin{cor}
We have $\varpi(\mc D_Q(u,\pa_u))=\big(\widetilde{\mc D}_{\mathsf Q}(u,\pa_u)\big)^{-1}$.	
\end{cor}

We then consider the image of the universal off-shell Bethe vectors \eqref{EQ:BAEU}, \eqref{eq:bae-sym} under the isomorphism $\varpi$. Let $\bm\xi=(\xi^1,\dots,\xi^{N-1})$ be a sequence of nonnegative integers, $$\bm t=(t_1^1,\dots,t^{1}_{\xi^1};\dots;t_1^{N-1},\dots,t^{N-1}_{\xi^{N-1}})$$ a sequence of variables. Set $\bar{\bm\xi}=(\xi^{N-1},\dots,\xi^1)$ and $\bar{\bm t}=(t^{N-1}_{\xi^{N-1}},\dots,t_1^{N-1};\dots;t^{1}_{\xi^1},\dots,t_1^1)$.

By the fact that $(R(u))^\dag=R(u)$ and Yang-Baxter equation \eqref{eq:Yang-Baxter},  we have 
\beq\label{eq:R-dag}
\Big(\mathop{\overrightarrow\prod}\limits_{(a,i)<(b,j)}R^{(\xi^{<b}+j,\xi^{<a}+i)}(t_j^b-t_i^a)\Big)^\dag = \mathop{\overrightarrow\prod}\limits_{(a,i)<(b,j)}R^{(\xi^{<b}+j,\xi^{<a}+i)}(t_j^b-t_i^a),
\eeq 
where $\dag$ is taken over all factors. 

\begin{lem}[cf. {\cite[Lemma 5.1]{PRS17}}]
The image of ${\bB}_{\bm\xi}(\bm t)$ under the isomorphism $\varpi$ equals to $\widetilde{\bB}_{\bar{\bm\xi}}(\bar{\bm t})$ up to sign. 
\end{lem}
\begin{proof}
Instead of working on ${\bB}_{\bm\xi}(\bm t)$, we apply $\varpi$ to $\widehat{\bB}_{{\bm\xi}}({\bm t})$. Note that $R(u)\in\End(\scrV^{\otimes 2})$  corresponds to $ -\widetilde R(-u)\in\End(\widetilde\scrV^{\otimes 2})$. In the following, we use the symbol $\propto$ to denote the proportionality up to signs. Then we have
\begin{align*}
	\varpi(\widehat{\bB}_{\bm\xi}(\bm t))\stackrel{\eqref{eq:pi-T}}{\propto}\ & (\str\otimes \mathrm{id})\Big(\widetilde{\mathsf T}^{(1,|\bm\xi|+1)}(t_1^1)\cdots \widetilde{\mathsf T}^{(|\bm\xi|,|\bm\xi|+1)}(t_{\xi^{N-1}}^{N-1})\\
&\times \mathop{\overrightarrow\prod}\limits_{(a,i)<(b,j)}\widetilde R^{(\xi^{<b}+j,\xi^{<a}+i)}(t_i^a-t_j^b)E_{2'1'}^{\otimes \xi^1}\otimes\cdots\otimes E_{N',(N-1)'}^{\otimes \xi^{N-1}}\otimes 1\Big)\\
\stackrel{\eqref{eq transpose cyclic}}{\propto}\ & (\str\otimes \mathrm{id})\Big(\big(E_{1'2'}^{\otimes \xi^1}\otimes\cdots\otimes E_{(N-1)',N'}^{\otimes \xi^{N-1}}\otimes 1\big) \\
&\times \mathop{\overrightarrow\prod}\limits_{(a,i)<(b,j)}\widetilde R^{(\xi^{<b}+j,\xi^{<a}+i)}(t_i^a-t_j^b)\widetilde{T}^{(1,|\bm\xi|+1)}(t_1^1)\cdots \widetilde{T}^{(|\bm\xi|,|\bm\xi|+1)}(t_{\xi^{N-1}}^{N-1})\Big),\\
\stackrel{\eqref{eq:cyc-tr}}{\propto}\ &(\str\otimes \mathrm{id})\Big(\mathop{\overrightarrow\prod}\limits_{(a,i)<(b,j)}\widetilde R^{(\xi^{<b}+j,\xi^{<a}+i)}(t_i^a-t_j^b)\widetilde{T}^{(1,|\bm\xi|+1)}(t_1^1)\cdots \\
& \times \widetilde{T}^{(|\bm\xi|,|\bm\xi|+1)}(t_{\xi^{N-1}}^{N-1})    E_{1'2'}^{\otimes \xi^1}\otimes\cdots\otimes E_{(N-1)',N'}^{\otimes \xi^{N-1}}\otimes 1\Big),
\\ \stackrel{\eqref{eq RTT}}{\propto}\ & (\str\otimes \mathrm{id})\Big(\widetilde{T}^{(|\bm\xi|,|\bm\xi|+1)}(t_{\xi^{N-1}}^{N-1})\cdots \widetilde{T}^{(1,|\bm\xi|+1)}(t_1^1)\\
& \times \mathop{\overrightarrow\prod}\limits_{(a,i)<(b,j)}\widetilde R^{(\xi^{<b}+j,\xi^{<a}+i)}(t_i^a-t_j^b)  E_{1'2'}^{\otimes \xi^1}\otimes\cdots\otimes E_{(N-1)',N'}^{\otimes \xi^  {N-1}}\otimes 1\Big),
\\ \propto\ \ &(-1)^{|\bm \xi|}(\str\otimes \mathrm{id})\Big(\widetilde{T}^{(1,|\bm\xi|+1)}(t_{\xi^{N-1}}^{N-1})\cdots \widetilde{T}^{(|\bm\xi|,|\bm\xi|+1)}(t_1^1)\\
& \times \mathop{\overrightarrow\prod}\limits_{(a,i)<(b,j)}\widetilde R^{(|\bm\xi|+1-\xi^{<b}-j,|\bm\xi|+1-\xi^{<a}-i)}(t_i^a-t_j^b)  E_{(N-1)',N'}^{\otimes \xi^{N-1}}\otimes\cdots\otimes E_{1'2'}^{\otimes \xi^1}\otimes 1\Big),
\end{align*}
where we applied conjugation by the operator in $\End(\widetilde \scrV^{\otimes k})$ which reverse the order of tensor factors. Now the statement follows from \eqref{EQ:BAEU} and \eqref{eq:bae-sym} for $\mathrm{Y}(\gl_{n|m})$.
\end{proof}

\begin{lem}
The isomorphism $\varpi$ sends $I_{\bm\xi,\bm t,Q}$ for $\mathrm{Y}(\glMN)$ to $\tilde I_{\bar{\bm\xi},\bar{\bm t},\mathsf Q}$ for $\mathrm{Y}(\gl_{n|m})$.
\end{lem}
\begin{proof}
It suffices to check that $\varpi$ maps $\mathscr I_{\bm\xi,Q}^{a,i}(\bm t)$ in \eqref{eq:I} for $\mathrm{Y}(\glMN)$ to $-\tilde{\mathscr I}_{\bar{\bm\xi},\mathsf Q}^{a',j}(\bar{\bm t})$ for $\mathrm{Y}(\gl_{n|m})$, where $j=\xi^a+1-i$. Observing that the sequence of polynomials associated to $\bar{\bm t}$ and $\bar{\bm\xi}$ is $\bar{\bm y}=(y_{N-1},\dots,y_1)$ and $\ka_a$ for the former case corresponds to $-\tilde\ka_{a'}$ for the later one, then the lemma is straightforward.\end{proof}

\subsection{Proof of Theorem \ref{thm:main-tech}}\label{app:A}
We prove Theorem \ref{thm:main-tech} by induction on $m$. We first establish the base case $m=0$. By \eqref{eq diff trans}, transfer matrices $\mathfrak T_{k,Q}(u)$ associated to symmetrizers can be expressed in terms of transfer matrices $\mathscr T_{l,Q}(u)$ associated to antisymmetrizers. Now Theorem \ref{thm:main-tech} for the base case ($m=0$) follows from the observations from Section \ref{sec:correspondence} by applying the isomorphism $\varpi$ to \cite[Theorem 5.4]{MTV06}.

For the rest of the induction process, the procedure is almost parallel to that of \cite[Section 11]{MTV06}, cf. \cite{BR08}. We provide the details for completeness.

Since $\sV=\C\bfv_1\oplus\sW$ and $\bfv_1$ is even, one obtains $\sV^{\wedge k}=(\bfv_1\wedge \sW^{\wedge (k-1)})\oplus \sW^{\wedge k}$, where the first summand is spanned by vectors of the form $\bfv_1\wedge \bfv_{a_1}\wedge\cdots\wedge \bfv_{a_{k-1}}$ with $2\lle a_1<\cdots<a_b\lle m <a_{b+1}\lle \cdots\lle a_{k-1}\lle N$, while the second one is spanned by vectors of the form $\bfv_{a_1}\wedge\cdots\wedge \bfv_{a_{k}}$ with $1\lle a_1<\cdots<a_b\lle m <a_{b+1}\lle \cdots\lle a_{k}\lle N$, both for various $b$. We also identify $\sW^{\wedge (k-1)}$ with $(\bfv_1\wedge \sW^{\wedge(k-1)})$ by $\bfx\mapsto \bfv_1\wedge \bfx$.

The R-matrix $R_{\wedge k,\wedge 1}(u)$, see Section \ref{sec:fusion}, as an operator on $\sV^{\wedge k}\otimes \sV$ preserves the subspaces
\[
(\bfv_1\wedge \sW^{\wedge(k-1)})\oplus\C\bfv_1,\quad \big((\bfv_1\wedge \sW^{\wedge(k-1)})\otimes \sW \big)\oplus(\sW^{\wedge k}\otimes \C\bfv_1),\quad \sW^{\wedge k}\otimes \sW.
\] 
Moreover, we have
\beq\label{eq:high-R}
\begin{split}
R_{\wedge k,\wedge 1}(u)\big|_{(\bfv_1\wedge \sW^{\wedge(k-1)})\oplus\C\bfv_1}	=u+1,\quad R_{\wedge k,\wedge 1}(u)\big|_{\sW^{\wedge k}\otimes \sW}	=R^{\langle \cN-1\rangle}_{\wedge k,\wedge 1}(u),\\
R_{\wedge k,\wedge 1}(u)\big|_{\big((\bfv_1\wedge \sW^{\wedge(k-1)})\otimes \sW \big)\oplus(\sW^{\wedge k}\otimes \C\bfv_1)}=\begin{pmatrix}
R^{\langle \cN-1\rangle}_{\wedge (k-1),\wedge 1}(u) & \tilde{\mathsf S}^t\\
\tilde{\mathsf S} & u
\end{pmatrix},\quad 
\end{split}
\eeq
where $\tilde{\mathsf S}((\bfv_1\wedge \bfx)\otimes\bfw)=(-1)^{|\bfx||\bfw|}(\bfw\wedge \bfx)\otimes \bfv_1$.

Regard $T(u)$ and $T^{\wedge k}(u)$, see \eqref{eq:T-wedge-K}, as matrices over the super Yangian $\YglMN$ and consider the block decomposition induced by the decompositions $\sV=\C\bfv_1\oplus\sW$ and $\sV^{\wedge k}=(\bfv_1\wedge \sW^{\wedge(k-1)})\oplus \sW^{\wedge k}$,
\beq\label{eq:T-decom}
T(u)=\begin{pmatrix}
A(u) & B(u)\\
C(u) & D(u)	
\end{pmatrix},\quad 
T^{\wedge k}(u)=\begin{pmatrix}
\widehat A(u) & \wt B(u)\\
\wt C(u) & \wt D(u)	
\end{pmatrix}.
\eeq
For example, $A(u)=T_{11}(u)$, $B(u)=\sum_{a=2}^NE_{1a}\otimes T_{1a}(u)$, $D(u)=\sum_{a,b=2}^NE_{ab}\otimes T_{ab}(u)$.

We use the following convenient notations. Denote by $\mathrm{HY}(L,M)$ the superspace $\mathrm{Hom}(L,M)\otimes \YglMN$ of matrices with noncommuting entries. Call $L$ the domain of those matrices. For instance, the coefficients of the series $B$, $D$, $\wt A$, $\wt B$, $\wt D$ belong to $\mathrm{HY}(\sW,\C)$, $\mathrm{HY}(\sW,\sW)$, $\mathrm{HY}(\sW^{\wedge(k-1)},\sW^{\wedge(k-1)})$, $\mathrm{HY}(\sW^{\wedge k},\sW^{\wedge(k-1)})$, $\mathrm{HY}(\sW^{\wedge k},\sW^{\wedge k})$. In particular, if $k=1$, then $\wt X(u)=X(u)$ for $X=A,B,C,D$.

Renormalize R-matrices,
\[
\overline R(u)=\frac{1}{u}R^{\langle \cN-1\rangle}(u),\quad \widetilde R(u)=\frac{1}{u+1}R_{\wedge k,\wedge 1}^{\langle \cN-1\rangle}(u),\quad \wt R(u)=\frac{1}{u}R_{\wedge k,\wedge 1}^{\langle \cN-1\rangle}(u)
\]
and define an even linear map 
$$
\mathsf S:\sW^{\wedge (k-1)}\otimes \sW\to \sW^{\wedge k},\quad \mathsf S(\bfx\otimes\bfw)=(-1)^{|\bfx||\bfw|}\bfw\wedge \bfx.
$$
We have equalities
\beq
\wt B^{[1]}(u)\wt B^{[2]}(v)=\frac{u-v}{u-v+1}\wt B^{[2]}(v)\wt B^{[1]}(v)\wt R^{(12)}(u-v)
\eeq
\beq
\wt A^{(1)}(u)\wt B^{[2]}(v)=\wt B^{[2]}(v)\wt A^{(1)}(u)\widetilde R^{(12)}(u-v-1)+\frac{1}{u-v}\wt B(u)\mathsf S^{[12]}A(v),
\eeq
\beq
\wt D^{(1)}(u)\wt B^{[2]}(v)=\wt B^{[2]}(v)\wt D^{(1)}(u)\widehat R^{(12)}(u-v)-\frac{1}{u-v}\mathsf S\wt B^{[1]}(u)D^{(2)}(v),
\eeq
in $\mathrm{HY}(\sW^{\wedge k}\otimes \sW, \sW^{\wedge (k-1)})$, $\mathrm{HY}(\sW^{\wedge (k-1)}\otimes \sW, \sW^{\wedge (k-1)})$, and $\mathrm{HY}(\sW^{\wedge k}\otimes \sW, \sW^{\wedge k})$, respectively, where the superscripts in brackets indicate which tensor factors are domains of the corresponding matrices. In particular, if $k=1$, we have
\beq\label{eq:BB}
 B^{[1]}(u) B^{[2]}(v)=\frac{u-v}{u-v+1} B^{[2]}(v)  B^{[1]}(v)\overline R^{(12)}(u-v)
\eeq
\beq
 A^{(1)}(u) B(v)= \frac{u-v-1}{u-v}B(v) A^{(1)}(u)+\frac{1}{u-v} B(u) A(v),
\eeq
\beq
 D^{(1)}(u) B^{[2]}(v)= B^{[2]}(v) D^{(1)}(u)\overline R^{(12)}(u-v)-\frac{1}{u-v}  B^{[1]}(u)D^{(2)}(v),
\eeq
see also \cite[Equations (4.6)-(4.8)]{BR08}.

Let $\check{R}(u)=(u+(-1)^{|2 |})^{-1}P^{\langle \cN-1\rangle}R^{\langle \cN-1\rangle}(u)$. For a function $f(u_1,\dots,u_r)$ with values in matrices with the domain $\sW^{\otimes r}$ and a simple permutation $(i,i+1)$, $1\lle i< r$, set
\beq\label{eq:sym-act}
{}^{(i,i+1)}f(u_1,\dots,u_r)=f(u_1,\dots,u_{i-1},u_{i+1},u_i,u_{i+2},\dots,u_r)\check{R}^{(i,i+1)}(u_i-u_{i+1}).
\eeq
Note that the matrix $\check{R}(u)$ satisfies $\check{R}(u)\check{R}(-u)=1$ and
\[
\check{R}^{(12)}(u-v)\check{R}^{(23)}(u)\check{R}^{(12)}(v)=\check{R}^{(23)}(v)\check{R}^{(12)}(u)\check{R}^{(23)}(u-v).
\]
Due to this, \eqref{eq:sym-act} extends to an action of the symmetric group $\fkS_r$ on functions $f(u_1,\dots,u_r)$ with values in matrices with the domain $\sW^{\otimes r}$, $f\mapsto {}^{\sigma}f$, $\sigma\in\fkS_r$. Thanks to \eqref{eq:BB}, the expression $B^{[1]}(u_1)\cdots B^{[r]}(u_r)$ is invariant under this action of $\fkS_r$.

In general, for a function $f(u_1,\dots,u_r)$ with values in matrices with the domain $\sW^{\otimes r}$, define
\[
{}^R\mathrm{Sym}_{r} f(u_1,\dots,u_r)=\sum_{\sigma\in \fkS_r}  {}^\sigma f(u_1,\dots,u_r).
\]
\begin{prop}[{\cite[Proposition 11.5]{MTV06}}]\label{prop11.5}
We have
\beq
\begin{split}
&\ \wt A^{(0)}(u) B^{[1]}(u_1)\cdots B^{[r]}(u_r)\\=&\ B^{[1]}(u_1)\cdots B^{[r]}(u_r)\wt A^{(0)}\widetilde{R}^{(0r)}(u-u_r-1)\cdots \widetilde{R}^{(01)}(u-u_1-1)\\
&+\frac{1}{(r-1)!}\wt B(u){}^R\mathrm{Sym}_{r}\Big(\frac{1}{u-u_1}\prod_{i=2}^r\frac{u_1-u_i-1}{u_1-u_i}\mathsf S^{[01]}B^{[2]}(u_2)\cdots B^{[r]}(u_r) A(u_1)\Big),	
\end{split}
\eeq
\beq
\begin{split}
&\ \wt D^{(0)}(u) B^{[1]}(u_1)\cdots B^{[r]}(u_r)\\=&\ B^{[1]}(u_1)\cdots B^{[r]}(u_r)\wt D^{(0)}\wt{R}^{(0r)}(u-u_r)\cdots \wt{R}^{(01)}(u-u_1)\\
&-\frac{1}{(r-1)!}\mathsf S\wt B^{[0]}(u){}^R\mathrm{Sym}_{r}\Big(\frac{1}{u-u_1}B^{[2]}(u_2)\cdots B^{[r]}(u_r)\\
&\qquad\ \ \ \quad\qquad\qquad\qquad\qquad \qquad\times D^{(1)}(u_1)\overline R^{(1r)}(u_1-u_r)\cdots \overline R^{(12)}(u_1-u_2)\Big),	
\end{split}
\eeq
where the tensor products are counted by $0,1,\dots,r$.\qed
\end{prop}

We also need the following statements. Note that formulas here are slightly different from those in \cite{MTV06} as we are using the opposite coproduct, see \eqref{eq:def-psi}.
\begin{lem}[{\cite[Lemma 11.6]{MTV06}}]\label{lem:A2}
	We have
\begin{align*}
	&D^{(0)}(u)\overline R^{(01)}(u-u_r)\cdots \overline R^{(0r)}(u-u_1)=\psi(u_1,\dots,u_r)\big(T^{\langle \cN-1\rangle}(u)\big),\\
	&\widehat D^{(0)}(u)\widehat R^{(01)}(u-u_r)\cdots \widehat R^{(0r)}(u-u_1)=\psi(u_1,\dots,u_r)\Big(\big(T^{\langle \cN-1\rangle}(u)\big)^{\wedge k}\Big).\qedd
\end{align*}
\end{lem}

\begin{lem}[{\cite[Lemma 11.7]{MTV06}}]\label{lem:A3}
For any $X\in \mathrm Y(\gl_{\cN-1})$, we have
\begin{align}
&A(u)\psi(u_1,\dots,u_r)(X)\simeq \psi(u_1,\dots,u_r)(X)A(u),\label{eq11.7-1}\\
&\widehat A^{(0)}(u)\widetilde R^{(01)}(u-u_r-1)\cdots \widetilde R^{(0r)}(u-u_1-1)\psi(u_1,\dots,u_r)(X)\label{eq11.7-2}\\
& \qquad\qquad\simeq  \prod_{i=1}^r\frac{u-u_i-1}{u-u_i}A(u)\psi(u_1,\dots,u_r)\Big(\big(T^{\langle \cN-1\rangle}(u-1)\big)^{\wedge k}X\Big).\nonumber
\end{align}
\end{lem}
\begin{proof}
Let $\mathrm Y_\times(\gl_{\cN})$ be the left ideal of $\mathrm Y(\gl_{\cN})$ generated by the coefficients of the series $T_{a1}(u)$ for $2\lle a\lle N$. Note that $\mathrm Y_\times(\gl_{\cN})$ is a subideal of $\mathrm Y_+(\gl_{\cN})$. It is clear from the definition relations \eqref{com relations} that for any $Z\in \mathrm Y(\gl_{\cN-1})$ and $C\in \mathrm Y_\times(\gl_{\cN})$, the coefficients of $[T_{11}^{\langle \cN\rangle}(u),\psi(Z)]$ and $C\psi(Z)$ belong to $\mathrm Y_\times(\gl_{\cN})$. Therefore \eqref{eq11.7-1} follows from the fact that $A(u)=T_{11}^{\langle \cN\rangle}(u)$ and $\psi(u_1,\dots,u_r)(X)\in \psi(Y(\gl_{\cN-1}))$.

It follows from \cite[Proposition 2 \& Remark 2.4]{MR14} that the coefficients of entries of the matrix $\widehat A(u)-A(u) D^{\wedge k-1}(u-1)$ are in $\mathrm Y_\times(\gl_{\cN})$. Therefore, by Lemma \ref{lem:A2}, we have
\begin{align*}
	&\ \widehat A^{(0)}(u)\widetilde R^{(01)}(u-u_r-1)\cdots \widetilde R^{(0r)}(u-u_1-1)\psi(u_1,\dots,u_r)(X)\\
	\simeq &\  A(u) \big(D^{\wedge k-1}(u-1)\big)^{[0]}\widehat R^{(01)}(u-u_r-1)\cdots \widehat R^{(0r)}(u-u_1-1)\psi(u_1,\dots,u_r)(X)\\
	\simeq &\  A(u) \psi(u_1,\dots,u_r)\Big(\big(T^{\langle \cN-1\rangle}(u-1)\big)^{\wedge k-1}\Big)\psi(u_1,\dots,u_r)(X)\prod_{i=1}^r\frac{u-u_i-1}{u-u_i}\\
\simeq &\  A(u) \psi(u_1,\dots,u_r)\Big(\big(T^{\langle \cN-1\rangle}(u-1)\big)^{\wedge k-1}X\Big)\prod_{i=1}^r\frac{u-u_i-1}{u-u_i}.
\end{align*}
Here we also used the fact  $\psi(u_1,\dots,u_r)$ is a homomorphism of superalgebras.
\end{proof}

Now we are ready to finish the proof of Theorem \ref{thm:main-tech-xxx}. Let $Q=\sum_{a=1}^N Q_aE_{aa}^{\langle \cN\rangle}\in\End(\sV)$ and $\overline Q=\sum_{a=1}^{N-1} Q_{a+1}E_{a+1,a+1}^{\langle \cN\rangle}\in\End(\sW)$. Set $\widetilde Q=\overline Q^{\wedge(k-1)}$ and $\widehat Q=\overline Q^{\wedge k}$. By the definition of transfer matrices, see \eqref{eq:def-transfer}, we have
\beq
\mathscr T_{k,Q}=Q_1\str_{\sW^{\wedge(k-1)}}(\widetilde Q\widehat A(u))+\str_{\sW^{\wedge k}}(\widehat Q\widehat D(u)).
\eeq
Set $r=\xi^1$ and $u_i=t_i^1$, $1\lle i\lle r$. We have
\begin{align}
&\mathscr T_{k,Q}(u)B^{[1]}(u_1)\cdots B^{[r]}(u_r)\label{eq:11.20}\\
=&\mathop{\overrightarrow\prod}_{1\lle i\lle r} B^{[i]}(u_i)\Big(Q_1(\str_{\sW^{\wedge(k-1)}}\otimes \mathrm{id}^{\otimes r})\big(\widetilde Q^{(0)}\widehat A^{(0)}(u)\mathop{\overleftarrow\prod}_{1\lle j\lle r}\widetilde R^{(0j)}(u-u_j-1)\big)
\nonumber\\
&\qquad\qquad\qquad\qquad\qquad +(\str_{\sW^{\wedge k}}\otimes \mathrm{id}^{\otimes r})\big(\widehat Q^{(0)}\widehat D^{(0)}(u)\mathop{\overleftarrow\prod}_{1\lle j\lle r}\widehat R^{(0j)}(u-u_j)\big)\Big)\nonumber\\
&+\frac{1}{(r-1)!}{}^R\mathrm{Sym}_{u_1,\dots,u_r}^{(1,\dots,r)}\Big[\frac{1}{u-u_1}\mathscr B_Q^{[1]}(u)\mathop{\overrightarrow\prod}_{2\lle i\lle r}B^{[i]}(u_i)\Big(Q_1\prod_{j=2}^r\frac{u_1-u_j-1}{u-u_j}A(u_1)\nonumber\\
&\qquad\qquad\qquad\qquad\qquad\qquad\qquad\qquad\qquad\qquad-\overline{Q}^{(1)}D^{(1)}(u_1)\mathop{\overleftarrow\prod}_{2\lle j\lle r}\overline R^{(1j)}(u_1-u_j)\Big)\Big],\nonumber
\end{align}
where $\mathscr B_Q (u)=(\str_{\sW^{\wedge(k-1)}}\otimes \mathrm{id})\big(\overline Q^{\wedge(k-1)}\widehat B(u)\mathsf S\big)$ and the tensor factors for the products under the traces are counted by $0,1,\dots,r$. Here we also used the equality
\[
(\str_{\sW^{\wedge k}}\otimes \mathrm{id})\big(\overline Q^{\wedge k}\mathsf S(\widehat B(u)\otimes\mathrm{id})\big)=\mathscr B_Q(u)\overline Q
\]
which follows from the supercyclicity of the supertrace and the formula 
\[
\overline Q^{\wedge k}\mathsf S=\mathsf S(\overline Q^{\wedge (k-1)}\otimes\overline Q).
\]
Note that $\mathscr B_Q(u)=B(u)$ if $k=1$.

For an expression $f(v)$ set $\underset{v=w}{\mathrm{res}}f(v)=\big((v-w)f(v)\big)\big|_{v=w}$ if the substitution makes sense. By Lemma \ref{lem:A2} and the equality $\underset{v=u_1}{\mathrm{res}}\big(\overline R(v-u_1)\big)=P^{\langle \cN-1\rangle}$, we have
\begin{align*}
\overline{Q}^{(1)}D^{(1)}(u_1)& \mathop{\overleftarrow\prod}_{2\lle j\lle r}\overline R^{(1j)}(u_1-u_j)\\
 =\ &\underset{v=u_1}{\mathrm{res}}\Big((\str_{\sW}\otimes\mathrm{id})
(\overline{Q}^{(0)}D^{(0)}(v)\mathop{\overleftarrow\prod}_{1\lle j\lle r}\overline R^{(0j)}(v-u_j))\Big)\\
=\ & \underset{v=u_1}{\mathrm{res}}\Big((\str_{\sW}\otimes\mathrm{id})
(\overline{Q}^{(0)}D^{(0)}(v)\mathop{\overleftarrow\prod}_{1\lle j\lle r}\overline R^{(0,r+1-j)}(v-u_j))\Big)\\
=\ &\underset{v=u_1}{\mathrm{res}}\Big(\psi(u_1,\dots,u_r)\big(\mathscr T_{1,\overline Q}^{\langle\cN-1\rangle}(v)\big)\Big).
\end{align*}
In the second equality, we used the super cyclicity of super trace which allows us to permute factors by conjugating the super flip operators $P^{(i,j)}$. The same will also be used in the sequel which we shall not write explicitly.

Therefore, for any $X\in \mathrm Y(\gl_{\cN-1})$, by Lemma \ref{lem:A3} and \eqref{eq:11.20}, we have
\begin{align}
&\mathscr T_{k,Q}(u)B^{[1]}(u_1)\cdots B^{[r]}(u_r)\psi(u_1,\dots,u_r)(X)\label{eq11.25}\\
\simeq &\mathop{\overrightarrow\prod}_{1\lle i\lle r} B^{[i]}(u_i)\Big(\psi(u_1,\dots,u_r)\big(\mathscr T_{k-1,\overline Q}^{\langle \cN-1\rangle}(u-1)X\big) Q_1A(u)\prod_{j=1}^r\frac{u-u_j-1}{u-u_j}\nonumber\\
&\qquad\qquad\qquad\qquad\qquad\qquad\qquad\qquad\qquad +\psi(u_1,\dots,u_r)\big(\mathscr T_{k,\overline Q}^{\langle \cN-1\rangle}(u)X\big)\Big)\nonumber\\
&+\frac{1}{(r-1)!}{}^R\mathrm{Sym}_{u_1,\dots,u_r}^{(1,\dots,r)}\Big[\frac{1}{u-u_1}\mathscr B_Q^{[1]}(u)\mathop{\overrightarrow\prod}_{2\lle i\lle r}B^{[i]}(u_i)\nonumber\\
& \times\Big(\psi(u_1,\dots,u_r)(X)Q_1A(u_1)\prod_{j=2}^r\frac{u_1-u_j-1}{u-u_j}-\underset{v=u_1}{\mathrm{res}}\big(\psi(u_1,\dots,u_r)\big(\mathscr T_{1,\overline Q}^{\langle\cN-1\rangle}(v)X\big)\big)\Big)\Big].	\nonumber
\end{align}

Then we apply both sides of \eqref{eq11.25} to the vector $\bfw_1^{\otimes r}$, set $X=\bB^{\langle\cN-1\rangle}_{\bar{\bm\xi}}(\bar{\bm t})$, see Proposition \ref{prop:BV-rec-XXX}, and employ the induction assumption. The first step amounts to replacing $\psi(u_1,\dots,u_r)$ by $\widetilde\psi(u_1,\dots,u_r)$, see Section \ref{sec:bv-recur}, and changing  ${}^R\mathrm{Sym}_{u_1,\dots,u_r}^{1,\dots,r}$ to the ordinary symmetrization $\mathrm{Sym}_{u_1,\dots,u_r}$ due to the fact that $\check{R}(u)\bfw_1\otimes \bfw_1=\bfw_1\otimes \bfw_1$.

We discuss the next two steps. Recall that $A(u)=T_{11}(u)$, $r=\xi^1$, and $u_i=t_i^1$, $1\lle i\lle r$. By Lemma \ref{lem:plus-proj}, \ref{lem11.1}, and
\[
\pi(x)(T_{aa}^{\langle\cN-1 \rangle}(u))=\Big(1+\frac{\ka_2\delta_{1a}}{u-x}\Big)\bfw_1,\quad 1\lle a\lle N-1,
\]
we have
\begin{equation}\label{eq:ka_2}
\widetilde\psi(t_1^1,\dots,t_{\xi^1}^1)(T_{aa}^{\langle\cN-1 \rangle}(u))=T_{a+1,a+1}(u) \otimes  \bfw_1^{\otimes r}\prod_{i=1}^{\xi^1}\Big(1+\frac{\ka_2\delta_{1a}}{u-t_i^1}\Big),\quad 1\lle a\lle N-1.
\eeq
Then \eqref{eq11.25} becomes
\begin{align}
&\ \mathscr T_{k,Q}(u)\bB_{\bm\xi}(\bm t)\label{eq11.26}\\
\simeq &\ \bB_{\bm\xi}(\bm t)\Big(Q_1T_{11}(u)\prod_{j=1}^{\xi^1}\frac{u-u_j-1}{u-u_j}\sum_{\bm a}\prod_{j=2}^{k}\ka_{a_j}\mathscr X_{\bm\xi,Q}^{a_j}(u-j+1;\bm t)\nonumber\\
&\qquad\qquad\qquad\qquad\qquad\qquad\ \  + \sum_{\bm b}\prod_{j=1}^{k}\ka_{b_j}\mathscr X_{\bm\xi,Q}^{b_j}(u-j+1;\bm t)\Big)+\mathscr U_{\bm \xi,k,Q}(u;\bm t)\nonumber
\nonumber
\end{align}
where the sums are taken over all possible $\bm a=(2\lle a_2<\cdots<a_i<m+1\lle a_{i+1}\lle \cdots\lle a_k\lle N)$ and $\bm b=(2\lle b_1<\cdots<b_j<m+1\lle b_{j+1}\lle \cdots\lle b_k\lle N)$ for various $1\lle i\lle k$ and $0\lle j\lle k$, respectively, and
\begin{align}
&\ \mathscr U_{\bm \xi,k,Q}(u;\bm t)\label{eq11.27}\\
\simeq &\ \mathop{\overrightarrow\prod}_{1\lle i\lle t^1_{\xi_1}} B^{[i]}(t_i^1)\Big(\widetilde\psi(t_1^1,\dots,t_{\xi_1}^1)\big(\mathscr U^{\langle \cN-1\rangle}_{\bar{\bm \xi},k-1,\overline Q}(u-1;\bar{\bm t})\big) Q_1T_{11}(u)\prod_{j=1}^{\xi^1}\frac{u-t_j^1-1}{u-t_j^1} \nonumber\\
&\qquad\qquad\qquad\qquad\qquad\qquad\qquad\qquad\quad\quad+ \widetilde\psi(t_1^1,\dots,t_{\xi^1}^1)\big(\mathscr U^{\langle \cN-1\rangle}_{\bar{\bm \xi},k ,\overline Q}(u;\bar{\bm t})\big)\Big) \nonumber
\nonumber\\
& + \frac{1}{(\xi^1-1)!}\mathrm{Sym}_{t_1^1,\dots,t_{\xi^1}^1}^{(1,\cdots,\xi^1)}\Bigg[\frac{1}{u-t_1^1}\mathscr B_Q^{[1]}(u)\mathop{\overrightarrow\prod}_{2\lle i\lle \xi^1}B^{[i]}(t_i^1)\nonumber\\
&\quad \times\bigg(\widetilde\psi(t_1^1,\dots,t_{\xi^1}^1)(\bB_{\bar{\bm\xi}}^{\langle \cN-1\rangle}(\bm t))\nonumber\\
&\quad\quad\times\Big(Q_1T_{11}(t_1^1)\prod_{j=2}^{\xi^1}\frac{t_1^1-t_j^1-1}{t_1^1-t_j^1}-\ka_2Q_2T_{22}(u)\frac{y_1(t_1^1+\ka_2)y_2(t_1^1-\ka_2)}{y_1'(t_1^1)y_2(t_1^1)}\Big)\nonumber
\\
&\qquad\qquad\qquad\qquad\qquad\qquad\qquad\qquad\qquad\qquad-\underset{v=u_1}{\mathrm{res}}\Big(\widetilde\psi(t^1_1,\dots,t_{\xi^1}^1)\big(\mathscr U^{\langle \cN-1\rangle}_{\bar{\bm \xi},1,\overline Q}(v;\bar{\bm t})\big) \Big)\bigg)\Bigg].\nonumber
\end{align}
Here $\ka_2$ in $y_1(t_1^1+\ka_2)$ comes from \eqref{eq:ka_2}. Finally, using the induction hypothesis that the expressions $\mathscr U^{\langle \cN-1\rangle}_{\bar{\bm \xi},j,\overline Q}(v;\bar{\bm t})$ are contained in $I^{\langle \cN-1\rangle}_{\bar{\bm \xi},j,\overline Q}(v;\bar{\bm t})$, Lemmas \ref{lem:plus-proj}, \ref{lem:antisym-sym-R}, \ref{lem11.1}, \ref{lem11.2}, and formulas \eqref{eq:products-I}, \eqref{eq:def:eva-2}, \eqref{eq:def-psi}, we conclude that there exists an element $\mathscr U_{\bm \xi,k,Q}(u;\bm t)$ in $I_{\bm\xi,\bm t,k,Q}$ satisfying \eqref{eq11.27}, completing the proof of Theorem \ref{thm:main-tech} from \eqref{eq11.26}.

\begin{proof}[Proof of Proposition \ref{prop:singular}] Let $C(u)=\sum_{a=2}^{N}E_{a1}\otimes T_{a1}(u)$ be the left bottom block of $T(u)$ in \eqref{eq:T-decom}. Let $C(u)=\sum_{s=1}^\infty C_su^{-s}$. Similar to the proof of \cite[Proposition 6.2]{MTV06}, we have
\[
C_1B(u)-B(u)C_1=\ka_1(A(u)-D(u))=A(u)-D(u).
\]
The rest is parallel to that of \cite[Proposition 6.2]{MTV06} which we shall skip the detail.
\end{proof}

\section{Gaudin models}\label{sec:Gaudin}
By taking the classical limits, we obtain the corresponding result for Gaudin models in this section. We start with preparing notations for the Gaudin case.
\subsection{Current superalgebra}%

Let  $\glMN[x]$ be the Lie superalgebra $\glMN\otimes\C[x]$ of $\glMN$-valued polynomials in $x$ with the point-wise supercommutator. Call $\glMN[x]$ the \emph{current superalgebra} of $\glMN$. 

We write $e_{ab}^{\{r\}}$ for $e_{ab}\otimes x^r$, $r\in \Z_{\gge 0}$. A basis of $\glMN[x]$ is given by $e_{ab}^{\{r\}}$, $1\lle a,b\lle N$ and $r\in \Z_{\gge 0}$. They satisfy the supercommutator relations
\[
[e_{ab}^{\{r\}},e_{cd}^{\{s\}}]=\delta_{bc}e_{ad}^{\{r+s\}}-(-1)^{(|a|+|b|)(|c|+|d|)}\delta_{ad}e_{cb}^{\{r+s\}}.
\]
We identify $\glMN$ with the subalgebra $\glMN\otimes 1$ of constant polynomials in $\glMN[x]$, that is we identify $e_{ab}$ in $\glMN$ with $e_{ab}^{\{0\}}$ in $\glMN[x]$. Denote by $\mathrm{U}(\glMN[x])$ the universal enveloping superalgebra of $\glMN[x]$. Let $\n_{\pm}[x]$ be the corresponding subalgebras in $\glMN[x]$, see \eqref{eq:nilpotent}.

We say that a vector in a $\glMN[x]$-module is called a \emph{weight singular vector} if $\n_+[x]v=0$ and $v$ is an eigenvector for all $e_{aa}^{\{s\}}$, $1\lle a\lle N$, $s\in\Z_{>0}$.

For any $\glMN$-module $M$ and $z\in \C$, we have the evaluation $\glMN[x]$-module $M\llparenthesis z\rrparenthesis$ at the evaluation point $z$ with the action given by $e_{ab}^{\{s\}}\big|_{M\llp z\rrp}=z^s e_{ab}\big|_{M}$. 

Consider the generating series and generating matrix 
\[
L_{ab}(u)=(-1)^{|b|}\sum_{s=0}^\infty e_{ba}^{\{s\}}u^{-s-1}\in \mathrm{U}(\glMN[x])[[u^{-1}]],\quad 1\lle a,b\lle N,
\]
\[
L(u)=\sum_{a,b=1}^N E_{ab}\otimes L_{ab}(u)\in \EndV\otimes \mathrm{U}(\glMN[x])[[u^{-1}]].
\]
Note that the indices in the generating series are flipped and the signs in the generating matrix are added so that it matches the evaluation map of super Yangian we used in \eqref{eq:eval}.

\subsection{Gaudin transfer matrices}
To define Gaudin transfer matrices, we first recall basics about pseudo-differential operators.

Let $\mathscr A$ be a differential superalgebra with an even derivation $\pa:\mathscr A\to\mathscr A$. For $r\in\Z_{>0}$, denote the $r$-th derivative of $a\in\mathscr A$ by $a_{[r]}$. Define the \emph{superalgebra of pseudo-differential operators} $\mathscr A((\pa^{-1}))$ as follows. Elements of $\mathscr A((\pa^{-1}))$ are Laurent series in $\pa^{-1}$ with coefficients in $\mathscr A$, and the product is given by
\[
\pa\pa^{-1}=\pa^{-1}\pa=1,\quad \pa^r a=\sum_{s=0}^\infty {r \choose s}a_{[s]}\pa^{r-s},\quad r\in\Z,\quad a\in\mathscr A,
\]
where 
$$
{r \choose s}=\frac{r(r-1)\cdots(r-s+1)}{s!}.
$$

Let 
$$
\mathscr A_{u}^{m|n}=\mathrm{U}(\glMN[x])((u^{-1}))=\Big\{ \sum_{r=-\infty}^s g_ru^{r},\ g_r\in \mathrm{U}(\glMN[x]), \ s\in \Z \Big\}.
$$
Fix a matrix $K=(K_{ab})_{1\lle a,b\lle N}\in\EndV$, then the operator in $\EndV\otimes \mathscr A_{u}^{m|n}((\pa_u^{-1}))$,
$$
\mathfrak Z_{K}(u,\pa_u):=\pa_u-K-L^\dag(u)=\sum_{a,b=1}^NE_{ab}\otimes\left(\delta_{ab}\pa_u-K_{ab}-L_{ab}(u)(-1)^{|a||b|+|a|+|b|}\right)
$$
is a Manin matrix, see \cite[Lemma 3.1]{MR14} and \cite[Lemma 4.2]{HM20}. 

Consider the quantum Berezinian $\mathrm{Ber}(\mathfrak Z_{K}(u,\pa_u))$ and expand it as an element in $\mathscr A_{u}^{m|n}((\pa_u^{-1}))$,
\beq\label{eq:Ber-Gaudin}
\mathfrak{D}_{K}(u,\pa_u)=\mathrm{Ber}(\pa_u-K-L^\dag(u))=\sum_{r=0}^{\infty}\mathcal G_{r,K}(u)\pa_u^{m-n-r}.
\eeq
We call the series $\mathcal G_{r,K}(u)\in \mathscr A_{u}^{m|n}$, $r\in \Z_{\gge 0}$, the \emph{Gaudin transfer matrices}. 

Note that this family of series are different from that in \cite{MR14}, see Section \ref{sec:more-G-Trans}. However, the coefficients of those two family of series generate the same subalgebra of $\mathrm{U}(\glMN[x])$ which we call the \emph{Bethe subalgebra} of $\mathrm{U}(\glMN[x])$, see \cite[Proposition 4.4]{HM20}.

The following properties about Gaudin transfer matrices are known.
\begin{lem}[\cite{MR14}] We have
\begin{enumerate}
\item $[\mathcal G_{r,K}(u),\mathcal G_{s,K}(v)]=0$;
\item if $K$ is the zero matrix, then the coefficients of $\mathcal G_{r,K}(u)$ commutes with the subalgebra $\UglMN$ of $\mathrm{U}(\glMN[x])$.
\end{enumerate}
\end{lem}

\subsection{Bethe vectors}
Let $M_1,\dots,M_\ell$ be $\glMN$-modules, $\bm z=(z_1,\dots,z_\ell)$ a sequence of complex numbers. Consider the tensor product of evaluation $\glMN[x]$-modules $ M\llp\bm z\rrp =M_1\llp z_1\rrp\otimes\cdots \otimes M_\ell\llp z_\ell\rrp$. Then we have
\[
L_{ab}(u)\big|_{M(\bm z)}=\ka_b\sum_{i=1}^\ell\frac{e_{ba}^{(i)}}{u-z_i}
\]
and
 the operator $\mathcal G_{r,K}^{M}(u;\bm z)=\mathcal G_{r,K}(u)|_{ M(\bm z)}$, for each $r\in\Z_{\gge 0}$, is a rational function in $u,\bm z$ with the denominator $\prod_{i=1}^\ell(u-z_i)^{r}$. We call the operators $\mathcal G_{r,K}^{M}(u;\bm z)$ the \emph{transfer matrices of the Gaudin model} on $M\llp\bm z\rrp$ associated with $\glMN$. 
 
We are interested in the case when $M_1,\dots, M_\ell$ are highest weight $\glMN$-modules with highest weights $\La_1,\dots,\La_\ell$, where $\La_i=(\La_i^1,\dots,\La_i^N)$, and highest weight vectors $v_1,\dots,v_\ell$. Set $\bLa=(\La_1,\dots,\La_\ell)$. In this case, the vector $ v^+=v_1\otimes\cdots\otimes v_{\ell}$ is a singular weight vector of $\mathrm{U}(\glMN[x])$ in $M\llp\bm z\rrp$,
\beq\label{eq:wt}
L_{aa}(u)v^+=v^+\kappa_a\ \Big(\sum_{i=1}^\ell\frac{\La_i^a}{u-z_i}\Big),\quad 1\lle a\lle N.
\eeq

Now we assume that $K=\sum_{a=1}^N K_aE_{aa}\in\EndV$ is diagonal and use similar notations as in Section \ref{sec:bv}. 

Let $\bm\xi=(\xi^1,\dots,\xi^{N-1})$ be a sequence of nonnegative integers. Consider an expression $\mathbb F_{\bm\xi}(\bm t)$ in $|\bm\xi|$ variables $\bm t=(t_1^1,\dots,t^{1}_{\xi^1},\dots,t_1^{N-1},\dots,t^{N-1}_{\xi^{N-1}})$ with coefficients in $\mathrm{U}(\glMN[x])$ which will be defined later in \eqref{eq:BV-gaudin-recur} of Section \ref{sec:bv-recur-Gaudin}. Here we only need to notice that $\mathbb F_{\bm\xi}(\bm t)$ is obtained from $\bB_{\bm\xi}(\bm t)$ by taking certain gradation, see Proposition \ref{prop:bv-gr}.
Apply $\mathbb F_{\bm\xi}(\bm t)$ to $v^+$ and renormalize it so that the function
\beq\label{eq:off-shell-bv-gaudin}
\mathbb F_{\bm \xi}^{v^+}(\bm t;\bm z)=\mathbb F_{\bm \xi}(\bm t)v^+\ \prod_{a=1}^{N-1}\prod_{i=1}^{\xi^a}\prod_{j=1}^{\ell}(t_i^a-z_\ell)\prod_{a=1}^{N-2}\prod_{i=1}^{\xi^a}\prod_{j=1}^{\xi^{a+1}}(t_j^{a+1}-t_i^a)
\eeq
is a polynomial in $\bm t$, $\bm z$, see Section \ref{sec:bv-recur-Gaudin}. We call $\mathbb F_{\bm \xi}^{v^+}(\bm t;\bm z)$ an \emph{off-shell Bethe vector} for the Gaudin model on $M\llp\bm z\rrp $ associated with $\glMN$.

Let $\bm y=(y_1,\dots,y_{N-1})$ be the sequence of polynomials associated to $\bm t$ and $\bm \xi$. The system of algebraic equations in $|\bm\xi|$ variables $\bm t$,
\beq\label{eq:Gaudinbae}
\begin{split}
K_a-K_{a+1}+\sum_{j=1}^\ell&\frac{\ka_a\La_j^a-\ka_{a+1}\La_j^{a+1}}{t_i^a-z_j}+\frac{\ka_ay_{a-1}'(t_i^a)}{y_{a-1}(t_i^a)}\\-&\frac{(\ka_a+\ka_{a+1})y_{a}''(t_i^a)}{2y_a'(t_i^a)}+\frac{\ka_{a+1}y_{a+1}'(t_i^a)}{y_{a+1}(t_i^a)}=0,
\end{split}
\eeq
$1\lle a\lle N-1$, $1\lle i\lle \xi^a$, is called the \emph{Bethe ansatz equation}, see \cite[Equation (4.2)]{MVY15}, which is usually written in the form,
\beq\label{eq:Gaudinbae2}
\mathfrak K_{\bm\xi,K}^{a,i}(\bm t;\bm z;\bLa)=0,
\eeq
where
\beq\label{eq:define-BAE-func}
\begin{split}
\mathfrak K_{\bm\xi,K}^{a,i}(\bm t;\bm z;\bLa)=K_a-K_{a+1}+\sum_{j=1}^\ell\frac{\ka_a\La_j^a-\ka_{a+1}\La_j^{a+1}}{t_i^a-z_j}+&\sum_{j=1}^{\xi^{a-1}}\frac{\ka_a}{t_i^a-t_{j}^{a-1}}\\-\sum_{j=1,j\ne i}^{\xi^{a}}\frac{\ka_a+\ka_{a+1}}{t_i^a-t_{j}^{a}}+&\sum_{j=1}^{\xi^{a+1}}\frac{\ka_{a+1}}{t_i^a-t_{j}^{a+1}},
\end{split}
\eeq
$1\lle a<N$, and $1\lle i\lle \xi^a$. We always assume that for a solution $\tl{\bm t}=(\tl t_1^1,\dots,\tl t_{\xi^{N-1}}^{N-1})$ of system \eqref{eq:Gaudinbae2}, any denominators in these equations does not vanish unless the corresponding numerator is zero. 

When $\tl{\bm t}$ is a solution of the Bethe ansatz equation \eqref{eq:bae}, we say that the vector $\mathbb F_{\bm \xi}^{v^+}(\tl{\bm t};\bm z)$ is an \emph{on-shell Bethe vector}.

\subsection{Main results for Gaudin models}

For $1\lle a\lle N$, define 
\beq\label{eq:chi-special-Gaudin}
\mathfrak X_{\bm\xi,K}^a(u;\bm t;\bm z;\bLa)=K_a+\ka_a\Big(\sum_{j=1}^\ell \frac{\La_j^a}{u-z_j}+\frac{y_{a-1}'(u)}{y_{a-1}(u)}-\frac{y_{a}'(u)}{y_{a}(u)}\Big),
\eeq
where $\bm y=(y_1,\dots,y_{N-1})$ is the sequence of polynomials associated to $\bm t$ and $\bm \xi$.

\begin{thm}\label{thm:main-tech-Gaudin}
Let $K$ be a diagonal matrix. If $M_1,\dots, M_\ell$ are highest weight $\glMN$-modules with highest weights $\La_1,\dots,\La_\ell$ and $\tl{\bm t}$ is an isolated solution of the Bethe ansatz equation \eqref{eq:Gaudinbae2}, then we have
\[
\mathfrak D_K(u,\pa_u )\mathbb F_{\bm\xi}^{v^+}(\tl{\bm t};\bm z)=\mathbb F_{\bm\xi}^{v^+}(\tl{\bm t};\bm z)\mathop{\overrightarrow\prod}\limits_{1\lle a\lle N}\big(\pa_u-\mathfrak X_{\bm\xi,K}^a(u;\tl{\bm t};\bm z;\bLa)\big)^{\ka_a},
\]	
where $\mathfrak D_K(u,\pa_u)$ is defined in \eqref{eq:Ber-Gaudin}.
\end{thm}

The theorem is proved in Section \ref{app:B}. It is an  analog of \cite[Theorem 3]{FFR94} for supersymmetric case in type A. Note that the condition for the solution being isolated in Theorem \ref{thm:main-tech-Gaudin} can be removed. If the solution is not isolated, then the statements can be proved similarly as in \cite[Theorems 8.6 \& 9.2]{MTV06} as we have done for the XXX spin chain case, cf. Section \ref{sec:cor-gaudin}. The theorem for the case of $m=n=1$ was announced in \cite[Theorem 4.11]{Lu22}.

Note that the pseudo-differential operator 
\beq\label{eq:pseudo-oper}
\mathfrak D_Q(u,\pa_u;\tl{\bm t};\bm z;\bLa):=\mathop{\overrightarrow\prod}\limits_{1\lle a\lle N}\big(\pa_u-\mathfrak X_{\bm\xi,K}^a(u;\tl{\bm t};\bm z;\bLa)\big)^{\ka_a}
\eeq
was introduced in \cite[Equation (6.5)]{HMVY19}.

\begin{prop}[{cf. \cite[Theorem 4.3]{MVY15}}]\label{prop:G-singular}
Let $K$ be the zero matrix. If $M_1,\dots, M_\ell$ are highest weight $\glMN$-modules with highest weights $\La_1,\dots,\La_\ell$ and $\tl{\bm t}$ is an isolated solution of the Bethe ansatz equation \eqref{eq:Gaudinbae2}, then the on-shell Bethe vector $\mathbb F_{\bm\xi}^{v^+}(\tl{\bm t};\bm z)$ is a $\glMN$-singular vector in $M_1\otimes\cdots\otimes M_\ell$ with weight 
$$
\Big(\sum_{i=1}^\ell\La_i^1-\xi^1,\sum_{i=1}^\ell\La_i^1+\xi^1-\xi^2,\dots,\sum_{i=1}^\ell\La_i^{N}+\xi^{N-1}\Big).
$$
\end{prop}

The proposition is proved in Section \ref{app:B}. 

\section{More on Gaudin models}\label{sec:Gaudin-proof}
\subsection{More on transfer matrices}\label{sec:more-G-Trans}
Recall that $\sN=m-n$ which is the supertrace of identity operator on $\scrV$ and also the super-dimension of $\scrV$. Note that here $\mathscr N$ may be negative. In the rest of this paper, our convention for ratios of factorials involving $\mathscr N$ is that we first assume $\mathscr N$ is a formal variable, then cancel common factors, and finally plug in $\mathscr N=m-n$.

\begin{lem}\label{lem:str-trasnfer}
For any $l\gge k$, and any distinct $i_1,\dots,i_k\in\{1,\dots,l\}$, we have
\begin{align*}
(\str_{\scrV^{\otimes l}}\otimes\mathrm{id})(Q^{(i_1)}\cdots Q^{(i_k)}T^{(i_1,l+1)}(u)\cdots T^{(i_k,l+1)}(u-k+1)\mathbb A_{\{l\}}^{(1\cdots l)})&\\ = \frac{k!(\mathscr N -k)!}{l!(\mathscr N -l)!}\mathscr T_{k,Q}(u)&.
\end{align*}
\end{lem}
\begin{proof}
If $l=k$ and $i_j=j$ for all $1\lle j\lle k$, then the statement is equivalent to \eqref{eq:def-transfer}. For general cases, it follows from the equalities
\[
P^{(ij)}Q^{(i)}T^{(i,l+1)}=Q^{(j)}T^{(j,l+1)}P^{(ij)},\quad P^{(ij)}\mathbb A_{\{l\}}=\mathbb A_{\{l\}}P^{(ij)},
\]
the cyclicity of supertrace, and the formula
\[
(\mathrm{id}^{\otimes k}\otimes \str_{\scrV^{\otimes (l-k)}})\mathbb A_{\{l\}}=\frac{k!(\mathscr N -k)!}{l!(\mathscr N -l)!}\mathbb A_{\{k\}}.\qedhere
\]
\end{proof}
\begin{rem}
Note that since $k\lle l$, the denominator would not be zero after cancellation. However, unlike the even case, the number $\dfrac{k!(\mathscr N -k)!}{l!(\mathscr N -l)!}$ can be zero for certain $l$ and $k\in \{1,\dots,l\}$. 	
\end{rem}

Define another family of difference operators in $\YglMN[[u^{-1},\pa_u]]$, cf. \cite{Tal06,MTV06}, 
\beq\label{eq:diff-L-XXX}
\mathcal D_{l,Q}(u,\pa_u)=(\str_{\scrV^{\otimes l}}\otimes\mathrm{id})\Big(\Big(\mathop{\overrightarrow\prod}\limits_{1\lle i\lle l} \big(1-Q^{(i)}T^{(i,l+1)}(u)e^{-\pa_u}\big)\Big)\mathbb A_{\{l\}}^{(1\cdots l)}\Big)
\eeq
for $l\in\Z_{>0}$. By Lemma \ref{lem:str-trasnfer}, we have the following corollary.
\begin{cor}\label{cor:diff}
For $l\in\Z_{>0}$, we have
\[
\mathcal D_{l,Q}(u,\pa_u)=\frac{1}{(\mathscr N-l)!}\sum_{k=0}^l(-1)^k\frac{(\mathscr N-k)!}{(l-k)!}\mathscr T_{k,Q}(u)e^{-k\pa_u}.
\]
\end{cor}

\begin{rem}
	It follows from \eqref{eq:Att} and $(\mathbb A_{\{l\}})^2=\mathbb A_{\{l\}}$ that 
\begin{align*}
\mathop{\overrightarrow\prod}\limits_{1\lle i\lle l} \big(1-Q^{(i)}T^{(i,l+1)}(u)e^{-\pa_u}\big)& \mathbb A_{\{l\}}^{(1\cdots l)}\\
= \mathbb A_{\{l\}}^{(1\cdots l)}	&\mathop{\overrightarrow\prod}\limits_{1\lle i\lle l} \big(1-Q^{(i)}T^{(i,l+1)}(u)e^{-\pa_u}\big) \mathbb A_{\{l\}}^{(1\cdots l)},
\end{align*}
see also \emph{\cite[Proposition 2.1]{MR14}}.\qed
\end{rem}

There are also another family of Gaudin transfer matrices, see \cite{MR14}, defined as follows. For each $l\in\Z_{>0}$, consider the formal differential operator,
\beq\label{eq:diff-L-Gaudin}
\mathfrak D_{l,K}(u,\pa_u)=(\str_{\scrV^{\otimes l}}\otimes\mathrm{id})\Big(\Big(\mathop{\overrightarrow\prod}\limits_{1\lle i\lle l} \big(\pa_u-K^{(i)}-L^{(i,l+1)}(u)\big)\Big)\mathbb A_{\{l\}}^{(1\cdots l)}\Big).
\eeq
Let $\mathfrak G_{lk,K}(u)\in \mathrm{U}(\glMN[x])[[u^{-1}]]$, $l\in\Z_{>0}$ and $1\lle k\lle l$, be the coefficients of $\mathfrak G_{l,K}(u,\pa_u)$,
\beq\label{eq:diff-L-Gaudin2}
\mathfrak D_{l,K}(u,\pa_u)=\sum^{l}_{k=0}(-1)^k \mathfrak G_{lk,K}(u)\pa_u^{l-k}.
\eeq
Let $w$ be a formal variable. It is known from \cite[Theorem 2.13]{MR14} that
\beq\label{eq:Ber-G-w}
\mathrm{Ber}(1+w\mathfrak Z_K(u,\pa_u))=\sum_{k=0}^\infty w^k\mathfrak D_{l,K}(u,\pa_u).
\eeq

\subsection{Filtration on $\YglMN$}\label{sec:filtration}
Consider a filtered superalgebra $A$ with an ascending filtration $\cdots\subset A_{s-1}\subset A_s\subset A_{s+1}\subset \cdots\subset A$. Denote by $\mathrm{gr}_{s}^A:A_s\to A_s/A_{s-1}$ the natural projection and identify the quotient spaces with the corresponding homogeneous subspaces in the graded superalgebra 
\[
\mathrm{gr}A=\bigoplus_{r\in\Z}A_r/A_{r-1}.
\]
Then $\mathrm{gr}_{s}^A$ is regarded as a map from $A_s$ to $\mathrm{gr}A$. We will simply write $\mathrm{gr}_{s}$ for $\mathrm{gr}_{s}^A$ when the superalgebra $A$ is clear in the context. The superalgebra $\EndV\otimes A$ also has a filtration induced from that on $A$.

The super Yangian $\YglMN$ has a degree function defined by $\deg T_{ab}^{\{s\}}=s-1$ for $1\lle a,b\lle N$ and $s\in\Z_{>0}$. Then $\YglMN$ is a filtered superalgebra with $\YglMN_s$ being the subspace spanned by elements whose degrees are at most $s$. It is well known that $\mathrm{gr}(\YglMN)=\mathrm{U}(\gl_{m|n}[x])$ and $\mathrm{gr}_{s-1}(T_{ab}^{\{s\}})=(-1)^{|b|}e_{ba}^{\{s\}}$.

The filtration on $\YglMN$ can be extended to the superalgebra $\YglMN[[u^{-1},\pa_u]]$: $\deg u^{-1}=\deg\pa_u=-1$. Clearly, $\mathrm{gr}(\YglMN[[u^{-1},\pa_u]])=\mathrm{U}(\gl_{m|n}[x])[[u^{-1},\pa_u]]$. The series $T_{ab}(u)-\delta_{ab}\in \YglMN[[u^{-1},\pa_u]]$ has degree $-1$ and
\beq\label{eq:gr-T}
\mathrm{gr}_{-1}(T_{ab}(u)-\delta_{ab})=L_{ab}(u),\quad \mathrm{gr}_{-1}(T(u)-1)=L(u).
\eeq

We assume further that $Q$ in \eqref{eq:def-transfer} is a series in $\EndV[[\zeta]]$ instead of simply in $\EndV$. Then transfer matrices are power series in $u^{-1}$ and $\zeta$ with coefficients in $\YglMN$. Hence we consider transfer matrices are elements in $\YglMN[[u^{-1},\pa_u,\zeta]]$. Results and construction adapting to this new assumption naturally generalizes to the described setting. Extend further the filtration on $\YglMN[[u^{-1},\pa_u]]$ to $\YglMN[[u^{-1},\pa_u,\zeta]]$ by $\deg \zeta=-1$. Similarly, in the Gaudin model case, we assume $K$ to be an element in $\EndV[[\zeta]]$.

By convention, set $\mathscr T_{0,Q}(u)=1$. For any $k\in\Z_{\gge 0}$, set
\[
\mathscr S_{k,Q}(u)=\frac{1}{(\sN-k)!}\sum_{i=0}^k(-1)^{k-i}\frac{(\sN-i)!}{(k-i)!}\mathscr T_{i,Q}(u).
\]
Equivalently, for $k,l\in\Z_{\gge 0}$, $\mathscr S_{k,Q}(u)$ satisfy
\beq\label{eq:new-trans}
\sum_{k=0}^l(-1)^k\frac{(\sN-k)!}{(l-k)!}\mathscr S_{k,Q}(u)y^{l-k}=\sum_{i=0}^l(-1)^i\frac{(\sN-i)!}{(l-i)!}\mathscr T_{i,Q}(u)(y+1)^{l-i},
\eeq
where $y$ is a formal variale.
\begin{prop}\label{prop:gr-operss}
Let $\deg(Q-1)\lle -1$ with $K=\mathrm{gr}_{-1}(Q-1)$. Then $\deg(\mathscr S_{k,Q}(u))=-k$ for all $k\in\Z_{\gge 0}$. Moreover, for any $l\in \Z_{\gge k}$, we have
\beq\label{eq:gr-image}
\mathrm{gr}_{-k}(\mathscr S_{k,Q}(u))=\mathfrak G_{kk,K}(u),\quad\frac{(\mathscr N-k)!}{(\mathscr N-l)!(l-k)!}
\mathrm{gr}_{-k}(\mathscr S_{k,Q}(u))=\mathfrak G_{lk,K}(u).
\eeq
In particular, we have
$$
\frac{(\mathscr N-k)!}{(\mathscr N-l)!(l-k)!}\mathfrak G_{kk,K}(u)=\mathfrak G_{lk,K}(u),
$$
\[
\mathfrak D_{l,K}(u,\pa_u)=\frac{1}{(\mathscr N-l)!}\sum_{k=0}^l(-1)^k\frac{(\mathscr N-k)!}{(l-k)!}\mathfrak G_{kk,K}(u)\pa_u^{l-k}.
\]
\end{prop}
\begin{proof}
Setting $y=e^{\pa_u}-1$ in \eqref{eq:new-trans}, by Corollary \ref{cor:diff}, we have
\beq\label{eq:s-ts11}
\sum_{k=0}^l(-1)^k\frac{(\sN-k)!}{(l-k)!}\mathscr S_{k,Q}(u)(e^{\pa_u}-1)^{l-k}=(\sN-l)!\mathcal D_{l,Q}(u,\pa_u)e^{l\pa_u}.	
\eeq
Then one has
\beq\label{eq:s-ts}
(\sN-k)!\mathscr S_{k,Q}(u)=\sum_{i=0}^k(-1)^i\frac{(\sN-i)!}{(k-i)!}\mathcal D_{i,Q}(u,\pa_u)e^{i\pa_u}(e^{\pa_u}-1)^{k-i}
\eeq
by the standard identity 
$$
\sum_{r=0}^s\frac{(-1)^r}{r!(s-r)!}=0,\quad s\gge 1.
$$
Note that $\deg(T(u)-1)=\deg(e^{\pa_u}-1)=-1$ and $\deg(Q-1)\lle -1$, we conclude from \eqref{eq:diff-L-XXX} that $\deg(\mathcal D_{i,Q}(u,\pa_u))=-i$. Therefore, $\deg(\mathscr S_{k,Q}(u))=-k$ by \eqref{eq:s-ts}.

Using $\mathrm{gr}_{-1}(e^{\pa_u}-1)=\pa_u$ and $\mathrm{gr}_{-1}(T(u)-1)=L(u)$, see \eqref{eq:gr-T}, and computing $\mathrm{gr}_{-l}(\mathcal D_{l,Q}(u,\pa_u))$ by \eqref{eq:s-ts11} and \eqref{eq:diff-L-XXX}, we obtain that
\[
(\sN-l)!\mathfrak D_{l,K}(u,\pa_u)=\sum_{k=0}^l(-1)^k\frac{(\sN-k)!}{(l-k)!}\mathrm{gr}_{-k}(\mathscr S_{k,Q}(u))\pa_u^{l-k},
\]
where we also used \eqref{eq:diff-L-Gaudin}. Now \eqref{eq:gr-image} follows from \eqref{eq:diff-L-Gaudin2}. The rests are now obvious.
\end{proof}

\begin{rem}Note that the last two formulas can also be proved by using similar methods used in Lemma \ref{lem:str-trasnfer} and Corollary \ref{cor:diff}.	\qed
\end{rem}

\subsection{Recurrence of Bethe vectors}\label{sec:bv-recur-Gaudin}
In this section, we define similar maps $$\psi,\qquad \psi(x_1,\dots,x_r),\qquad \tilde\psi(x_1,\dots,x_r)$$ for $\mathrm{U}(\glMN[x]).$ We shall use the same notations for the counterparts in Section \ref{sec:bv-recur}.

Define the embedding $\psi:\mathrm{U}(\gl_{\cN-1}[x])\hookrightarrow\mathrm{U}(\gl_{\cN}[x])$ by the rule \[
\psi(L_{ab}^{\langle\cN-1\rangle}(u))=L_{a+1,b+1}(u),\quad  1\lle a,b\lle N-1.
\]Define a map $\psi(x_1,\dots,x_r):\mathrm{U}(\gl_{\cN-1}[x])\to \mathrm{U}(\gl_{\cN}[x]) \otimes \End(\sW^{\otimes r})$ by
\[
\begin{split}
\psi(x_1,\dots,x_r)(L_{ab}^{\langle\cN-1\rangle})=\ &L_{a+1,b+1}(u)\otimes 1^{\otimes r}\\
&\quad +\sum_{i=1}^r 1\otimes \frac{1^{\otimes(i-1)}\otimes E_{ba}^{\langle\cN-1\rangle}\otimes 1^{\otimes(r-i)} }{u-x_{r+1-i}}.
\end{split}
\]
Define a map $\widetilde \psi:\mathrm{U}(\gl_{\cN-1}[x])\to \mathrm{U}(\gl_{\cN}[x])\otimes\sW^{\otimes r}$ by
\[
\widetilde \psi(x_1,\dots,x_r)=\widetilde \psi(x_1,\dots,x_r)= \psi(x_1,\dots,x_r)( 1\otimes\bfw_1^{\otimes r}).
\]
The following lemmas are straightforward.
\begin{lem}
	We have $\widetilde \psi(x_1,\dots,x_r)(\mathrm{U}(\gl_{\cN-1}[x]\n_\pm^{\langle\cN-1\rangle}[x])\subset \mathrm{U}(\gl_{\cN}[x]\n_\pm[x])\otimes\sW^{\otimes r}$.
\end{lem}

Similarly, define the embedding $\phi:\mathrm{U}(\gl_{\cN-2}[x])\hookrightarrow \mathrm{U}(\gl_{\cN-1}[x])$ by the rule \[
\phi(L_{ab}^{\langle\cN-2\rangle}(u))=L_{a+1,b+1}^{\langle\cN-1\rangle}(u), \quad 1\lle a,b\lle N-2. 
\]

\begin{lem}
	We have $\widetilde \psi(x_1,\dots,x_r)\circ\phi = (\psi \circ \phi)\otimes\bfw_1^{\otimes r}$.
\end{lem}

Recall that $\bar{\bm\xi}=(\xi^2,\dots,\xi^{N-1})$ and $\bar{\bm t}=(t_1^2,\dots,t_{\xi^2}^2;\dots;t_1^{N-1},\dots,t_{\xi^{N-1}}^{N-1})$. Define $\mathbb F_{\bm\xi}(\bm t)$ inductively by
\beq\label{eq:BV-gaudin-recur}
\mathbb F_{\bm\xi}(\bm t)=F^{(1)}(t_1^1)\cdots F^{(\xi^1)}(t^1_{\xi^1})\widetilde \psi(t_1^1,\dots,t^1_{\xi^1})({\mathbb F}^{\langle\cN-1\rangle}_{\bar{\bm\xi}}\big(\bar{\bm t})\big),
\eeq
where $F(u)=(L_{12}(u),\dots,L_{1N}(u))=\sum_{a=1}^{N-1} E_{1,a+1}\otimes L_{1,a+1}(u)$ and its coefficients are treated as elements in $\mathrm{Hom}(\sW,\C)\otimes \mathrm U(\glMN[x])$. 

Recall from Section \ref{sec:filtration} the degree function $\deg$ defined on $\YglMN[[u^{-1},\pa_i,\zeta]]$ by $\deg u^{-1}=\deg \pa_u=\deg \zeta=-1$. Extend the degree function to rational expressions in $\bm t$ with coefficients in $\YglMN[[u^{-1},\pa_i,\zeta]]$ by setting $\deg t_i^a=1$ and $\deg(t_i^a-t_j^b)^{-1}=-1$ for all possible $a,b,i,j$. Note that the maps $\psi$, $ \psi(t_1^1,\dots,t^1_{\xi_1})$, and $\widetilde \psi(t_1^1,\dots,t^1_{\xi_1})$ respect the degree function and the projections to the associated graded superalgebras. For instance, if $X\in\YglMN$ and $\deg X=k$, then $\deg\psi(X)=k$ and $\psi(\mathrm{gr}_kX)=\mathrm{gr}_k(\psi(X))$.
\begin{prop}\label{prop:bv-gr}
We have $\deg(\widetilde \bB_{\bm\xi}(\bm t))=-|\bm \xi|$ and $\mathrm{gr}_{-|\bm \xi|}(\widetilde \bB_{\bm\xi}(\bm t))=\mathbb F_{\bm \xi}(\bm t)$.	
\end{prop}
\begin{proof}
The statements follow from Proposition \ref{prop:BV-rec-XXX}, the equality \eqref{eq:gr-T}, and the definition of $\mathbb F_{\bm \xi}(\bm t)$ by induction.
\end{proof} 

\subsection{Proof of Theorem \ref{thm:main-tech-Gaudin}}\label{app:B}
We show Theorem \ref{thm:main-tech-Gaudin} by taking the classical limits of Corollary \ref{cor:main-tech-xxx}. We start with recalling the objects we would like to compare between XXX spin chains and Gaudin models.

Let $M_1,\dots,M_{\ell}$ be highest weight $\glMN$-modules with highest weights $\La_1,\dots,\La_\ell$, $\bm z=(z_1,\dots,z_{\ell})$ a sequence of complex numbers. Recall the tensor product of evaluation $\YglMN$-modules $M(\bm z)= M_1(z_1)\otimes \cdots \otimes M_\ell(z_\ell)$, the rational difference operator $\mc D_Q(u,\pa_u)$ defined by quantum Berezinian, see \eqref{eq:Ber-XXX} and \eqref{eq diff trans}, and the off-shell Bethe vector $\bB_{\bm \xi}^{v^+}(\bm t;\bm z)$, see \eqref{eq:off-shell-bv}. Consider the rational functions
\begin{align*}
\mathscr Q_{\bm\xi,Q}^{a,i}(\bm t;\bm z;\bLa)=\frac{Q_a}{ Q_{a+1}}&\prod_{j=1}^{\ell}\frac{t_i^a-z_j+\ka_a\La_j^a}{t_i^a-z_j+\ka_{a+1}\La_j^{a+1}}\prod_{j=1}^{\xi^{a-1}}\frac{t_i^a-t_j^{a-1}+\ka_a}{t_i^a-t_j^{a-1}}\\ \times \ &\prod_{j=1,j\ne i}^{\xi^{a}}\frac{t_i^a-t_j^{a}-\ka_a}{t_i^a-t_j^{a}+\ka_{a+1}}\prod_{j=1}^{\xi^{a+1}}\frac{t_i^a-t_j^{a+1}}{t_i^a-t_j^{a+1}-\ka_{a+1}},
\end{align*}
for $1\lle a<N$ and $1\lle i\lle \xi^a$, cf. the Bethe ansatz equation \eqref{eq:bae}, the rational functions $\mathscr X_{\bm\xi,Q}^a(u;\bm t;\bm z;\bLa)$, see \eqref{eq:chi-special},
 and the rational difference operator $\mc D_Q(u,\pa_u; \bm t;\bm z;\bLa)$, see \eqref{eq:rational-diff-oper}, which encodes the eigenvalues of transfer matrices for XXX spin chains, see Corollary \ref{cor:main-tech-xxx}.
 
Similarly, we have the corresponding objects for Gaudin models.  Recall the tensor product of evaluation $\glMN[x]$-modules $ M\llp \bm z\rrp= M_1\llp z_1\rrp\otimes \cdots \otimes M_\ell \llp z_\ell\rrp$, the pseudo-differential operator $\mathfrak D_K(u,\pa_u)$ defined by quantum Berezinian, see \eqref{eq:Ber-Gaudin}, and the off-shell Bethe vector $\mathbb F_{\bm \xi}^{v^+}(\bm t;\bm z)$. Consider the rational functions $\mathfrak K_{\bm\xi,K}^{a,i}(\bm t;\bm z;\bLa)$, see \eqref{eq:define-BAE-func}, the rational functions $\mathfrak X_{\bm\xi,Q}^a(u;\bm t;\bm z;\bLa)$, see \eqref{eq:chi-special-Gaudin},
 and the pseudo-differential operator $\mathfrak D_K(u,\pa_u; \bm t;\bm z;\bLa)$, see \eqref{eq:pseudo-oper}.

The objects associated to Gaudin models can be obtained from the corresponding objects for the XXX spin chains by taking the following limit.

Note that $ M(\bm z)$ and $ M\llp \bm z\rrp $ share the same space which we denote by $M$. Then the following operators
\[
T_{ab}^{M}(u;\bm z):=T_{ab}(u)\big|_{ M(\bm z)},\quad L_{ab}^{ M}(u;\bm z):=L_{ab}(u)\big|_{ M\llp \bm z\rrp },
\]
\[
\mc D^{ M}_{Q}(u,\pa_u;\bm z)=\mc D_{Q}(u,\pa_u)\big|_{M(\ve^{-1}\bm z)},\quad \mathfrak D^{M}_{Q}(u,\pa_u;\bm z):=\mathfrak D_{Q}(u,\pa_u)\big|_{M\llp \bm z\rrp }
\]
can be regarded as operators on ${M}$ depending on the corresponding parameters.

Set $\varepsilon^{-1}\bm z=(\ve^{-1}z_1,\dots,\ve^{-1}z_\ell)$ and $\ve^{-1}\bm t=(\ve^{-1}t_1^1,\dots,\ve^{-1}t_{\xi^{N-1}}^{N-1})$.

\begin{prop}
Let $Q=1+\ve K$. As $\ve \to 0$, we have
\beq\label{eq:asym1}
T_{ab}^{M}(\ve^{-1}u;\ve^{-1}\bm z)=\delta_{ab}+\ka_b\ve L_{ab}^{ M}(u;\bm z)+O(\ve^2),
\eeq
\beq\label{eq:asym4}
\mathscr Q_{\bm\xi,Q}^{a,i}(\ve^{-1}\bm t;\ve^{-1}\bm z;\bLa)=1+\ve \mathfrak K_{\bm\xi,K}^{a,i}(\bm t;\bm z;\bLa)+O(\ve^2), 
\eeq
\beq\label{eq:asym2}
\mc D^{M}_{Q}(\ve^{-1}u,\ve\pa_u;\ve^{-1}\bm z)=\ve^{m-n}\mathfrak D^{M}_{Q}(u,\pa_u;\bm z)+O(\ve^{m-n+1}),
\eeq
\beq\label{eq:asym6}
\mc D_{l,Q}(\ve^{-1}u,\ve \pa_u; \ve^{-1}\bm t;\ve^{-1}\bm z;\bLa)= \ve^{l}\mathfrak D_{l,K}(u,\pa_u; \bm t;\bm z;\bLa)+O(\ve^{l+1}),
\eeq	
\beq\label{eq:asym5}
\mc D_Q(\ve^{-1}u,\ve \pa_u; \ve^{-1}\bm t;\ve^{-1}\bm z;\bLa)= \ve^{m-n}\mathfrak D_K(u,\pa_u; \bm t;\bm z;\bLa)+O(\ve^{m-n+1}),
\eeq
\beq\label{eq:asym3}
\begin{split}
\bB_{\bm \xi}^{v^+}(\ve^{-1}\bm t;\ve^{-1}\bm z) \prod_{i=1}^\ell \prod_{a=1}^{N-1}\prod_{j=1}^{\xi^a} \frac{\ve}{t_j^a -z_i} \prod_{a=1}^{N-2}& \ \prod_{i=1}^{\xi^a}\prod_{j=1}^{\xi^{a+1}}\frac{\ve}{t_j^{a+1}-t_i^a}\\= &\ \ve^{|\bm\xi|}\mathbb F_{\bm \xi}^{v^+}(\bm t;\bm z)+O(\ve^{|\bm\xi|+1}).
\end{split}
\eeq
\end{prop}
\begin{proof}
The equalities \eqref{eq:asym1} and \eqref{eq:asym4} are straightforward.	The formulas \eqref{eq:asym2}, \eqref{eq:asym6}, \eqref{eq:asym5} are proved similarly as in Theorem	\ref{prop:gr-operss}. The formula \eqref{eq:asym3} essentially follows from Proposition \ref{prop:bv-gr}.
\end{proof}

\subsection{Correspondence between $\mathrm{U}(\glMN[x])$ and $\mathrm{U}(\glNM[x])$}\label{sec:cor-gaudin}

In this section, we discuss the symmetry between Gaudin models for $\glMN$ and $\gl_{n|m}$, cf. Section \ref{sec:correspondence}. We shall use similar conventions as in Section \ref{sec:correspondence}.

We have the following isomorphism
\[
\vartheta:\mathrm{U}(\glMN[x])\to \mathrm{U}(\glNM[x]),\quad L_{ab}(u)\mapsto \widetilde{L}_{b'a'}(u)(-1)^{|a'||b'|+|b'|}
\]
For each $l\in\Z_{>0}$, consider another formal differential operator,
\beq\label{eq:diff-L-Gaudin-S}
\mathbb D_{l,K}(u,\pa_u)=(\str_{\scrV^{\otimes l}}\otimes\mathrm{id})\Big(\Big(\mathop{\overrightarrow\prod}\limits_{1\lle i\lle l} \big(\pa_u-K^{(i)}-L^{(i,l+1)}(u)\big)\Big)\mathbb H_{\{l\}}^{(1\cdots l)}\Big).
\eeq
It is known from \cite[Theorem 2.13]{MR14} that
\beq\label{eq:Ber-G-w-S}
\big(\mathrm{Ber}(1+w\mathfrak Z_K(u,\pa_u))\big)^{-1}=\sum_{k=0}^\infty(-1)^kw^k\mathbb D_{l,K}(u,\pa_u).
\eeq

We have the following analogous results whose proofs are similar to that of the Yangian case. For a diagonal matrix $K=\sum_{a=1}^N K_aE_{aa}\in\EndV$, set $\mathsf K=\sum_{a=1}^N K_aE_{a'a'}\in\End(\widetilde\scrV)$.
\begin{lem}
We have	
\[
\vartheta(\mathfrak D_{l,K}(u,\pa_u))=(-1)^l\widetilde{\mathbb D}_{l,\mathsf K}(u,\pa_u),\quad \vartheta(\mathbb D_{l,K}(u,\pa_u))=(-1)^l\widetilde{\mathfrak D}_{l,\mathsf K}(u,\pa_u).
\]
\end{lem}
\begin{proof}
The lemma follows immediately from our identification of operators on $\scrV$ and $\widetilde{\scrV}$, and the fact that supertranspose respects the supertrace.
\end{proof}
\begin{cor}
We have $\vartheta(\mathrm{Ber}(1+w\mathfrak Z_K(u,\pa_u)))=\big(\mathrm{Ber}(1+w\widetilde{\mathfrak Z}_{\mathsf K}(u,\pa_u))\big)^{-1}$.	
\end{cor}

\begin{cor}
We have $\vartheta(\mathrm{Ber}(\mathfrak Z_K(u,\pa_u)))=\big(\mathrm{Ber}(\widetilde{\mathfrak Z}_{\mathsf K}(u,\pa_u))\big)^{-1}$.	
\end{cor}
\begin{proof}
Recall that $\mathscr A_{u}^{m|n}=\mathrm{U}(\glMN[x])((u^{-1}))$. Let $\mathscr A_{u}^{m|n}[\pa_u]\subset \mathscr A_{u}^{m|n}((\pa_u^{-1}))$ be the subalgebra of differential operators,
\[
\mathscr A_{u}^{m|n}[\pa_u]=\Big\{\sum_{i=0}^r a_i\pa_u^i,\ r\in\Z_{\gge 0},\ a_i\in \mathscr A_{u}^{m|n}\Big\}.
\]
Define a linear map $\Phi^{m|n}:\mathscr A_{u}^{m|n}((\pa_u^{-1}))\to \mathscr A_{u}^{m|n}[\pa_u]$,
\[
\Phi^{m|n}:\sum_{i=-\infty}^r a_i\pa_u^i\to \sum_{i=-\infty}^r a_i(w^{-1}+\pa_u)^i,
\]
where the right hand side is expanded by the rule: $(w^{-1}+\pa_u)^i=\sum_{j=0}^\infty {i \choose j}\pa_u^jw^{-i+j}$. Then the map $\Phi^{m|n}$ is an injective homomorphism of superalgebras, see \cite[Lemma 4.1]{HM20}. It is also clear that $\vartheta\circ\Phi^{m|n}=\Phi^{n|m}\circ \vartheta$. By the proof of \cite[Proposition 4.4]{HM20}, we have
\begin{align*}
w^{m-n}\Phi^{n|m}(\vartheta(\mathrm{Ber}(\mathfrak Z_K(u,\pa_u))))= &\ \vartheta(w^{m-n}\Phi^{m|n}(\mathrm{Ber}(\mathfrak Z_K(u,\pa_u))))\\
= &\  \vartheta(\mathrm{Ber}(1+w\mathfrak Z_K(u,\pa_u)))\\
= &\ \big(\mathrm{Ber}(1+w\widetilde{\mathfrak Z}_{\mathsf K}(u,\pa_u))\big)^{-1}\\
= &\ \big(w^{n-m}\Phi^{n|m}(\mathrm{Ber}(\widetilde{\mathfrak Z}_{\mathsf K}(u,\pa_u))\big)^{-1}\\
=&\ w^{m-n}\Phi^{n|m}( (\mathrm{Ber}(\widetilde{\mathfrak Z}_{\mathsf K}(u,\pa_u)))^{-1}).
\end{align*}
Now the claim follows from the injectivity of $\Phi^{n|m}$.
\end{proof}

\begin{lem} 
The image of ${\mathbb F}_{\bm\xi}(\bm t)$ under the isomorphism $\vartheta$ equals to $\widetilde{\mathbb F}_{\bar{\bm\xi}}(\bar{\bm t})$ up to sign. 
\end{lem}
\begin{proof}
The statement follows from the fact that $\varpi$ and $\vartheta$ preserve the degrees and the equality 
$$
\mathrm{gr}_{-|\bm\xi|}^{n|m}\circ \varpi=\vartheta\circ \mathrm{gr}_{-|\bm\xi|}^{m|n},
$$
where $\mathrm{gr}_s^{m|n}$ and $\mathrm{gr}_s^{n|m}$ are the graded maps for $\YglMN$ and $\mathrm{Y}(\gl_{n|m})$, respectively.
\end{proof}

\bibliographystyle{JHEP}

\begin{thebibliography}{0}

\bibitem{MR14} 
 A. Molev and E. Ragoucy,
 \emph{The MacMahon Master Theorem for right quantum superalgebras and higher Sugawara operators for $\widehat{\mathfrak{gl}}(m|n)$},
\textit{ Moscow Math. J.}, {\bf 14} (2014) 1, 83--119, 
  [\href{https://arxiv.org/abs/0911.3447}{{\ttfamily arXiv:0911.3447}}].

\bibitem{NO}
M. Nazarov and G. Olshanski, {\it Bethe subalgebras in twisted Yangians},
\textit{Commun. Math. Phys. }\textbf{178} (1996), 483--506, [\href{https://arxiv.org/abs/q-alg/9507003}{{\ttfamily q-alg/9507003}}].

\bibitem{BR08}
S. Belliard and E. Ragoucy, 
\emph{The nested Bethe ansatz for `all' closed spin chains}, 
\textit{J. Phys. A} \textbf{41} (2008), 295202, 
  [\href{https://arxiv.org/abs/0804.2822}{{\ttfamily arXiv:0804.2822}}].

\bibitem{STF79}
E. Sklyanin, L. Takhtadzhyan and L. Faddeev,
{\it Quantum inverse problem method. I}, \textit{Theor. Math. Phys.} \textbf{40} (1979), 688, [\href{https://inspirehep.net/literature/141835}{{{\scriptsize IN}SPIRE}}].

\bibitem{TF79}
L. Takhtadzhyan and L. Faddeev,
{\it The quantum method of the inverse problem and the Heisenberg XYZ model},  \textit{Russian Math. Surveys} \textbf{34} (1979) 11, [\href{https://inspirehep.net/literature/143417}{{{\scriptsize IN}SPIRE}}].

\bibitem{KR83} 
P. Kulish and N. Yu. Reshetikhin, 
{\it Diagonalization of $\mathrm{GL}(N)$ invariant transfer-matrices and quantum $N$-wave system (Lee model)}, \textit{J. Phys. A} \textbf{15} (1983), L591--L596, [\href{https://inspirehep.net/literature/196942}{{{\scriptsize IN}SPIRE}}]. 

\bibitem{MTV06} 
E. Mukhin, V. Tarasov and A. Varchenko,
{\it Bethe eigenvectors of higher transfer matrices},
\textit{J. Stat. Mech. Theor. Exp.} (2006) P08002,
  [\href{https://arxiv.org/abs/math/0605015}{{\ttfamily math/0605015}}].

\bibitem{Tsu97}
Z. Tsuboi,
{\it Analytic Bethe ansatz and functional equations for Lie superalgebra $\mathfrak{sl}(r+1|s+1)$}, \textit{J. Phys. A} \textbf{30} (1997), no. 22, 7975,
  [\href{https://arxiv.org/abs/0911.5387}{{\ttfamily arXiv:0911.5387}}].

\bibitem{LM21a}
K. Lu and E. Mukhin,
\emph{On the supersymmetric XXX spin chains associated with $\gl_{1|1}$}, \textit{Commun. Math. Phys.} \textbf{386} (2021), 711--747,
  [\href{https://arxiv.org/abs/1910.13360}{{\ttfamily arXiv:1910.13360}}].

\bibitem{HM20}
C.-L. Huang and E. Mukhin,
\emph{The duality of $\glMN$ and $\gl_k$  Gaudin models},
\textit{J. Algebra} {\bf 548} (2020), 1--24,
  [\href{https://arxiv.org/abs/1904.02753}{{\ttfamily arXiv:1904.02753}}].

\bibitem{HMVY19}
C.-L. Huang, E. Mukhin, B. Vicedo and C. Young,
\emph{The solutions of  $\glMN$ Bethe ansatz equation and rational pseudodifferential operators}, \textit{Sel. Math. New Ser.}  (2019) \textbf{25}:52,
  [\href{https://arxiv.org/abs/1809.01279}{{\ttfamily arXiv:1809.01279}}].

\bibitem{HLM19}
C.-L. Huang, K. Lu and E. Mukhin,
\emph{Solutions of $\glMN$ XXX Bethe ansatz equation and rational difference operators}, \textit{J. Phys. A} \textbf{52} (2019), 375204, 31pp,
  [\href{https://arxiv.org/abs/1811.11225}{{\ttfamily arXiv:1811.11225}}].

\bibitem{Mol21}
A. Molev,
\textit{Odd reflections in the Yangian associated with $\gl(m|n)$}, \textit{Lett. Math. Phys.} (2022) \textbf{112}:8, 15pp,
  [\href{https://arxiv.org/abs/2109.09462}{{\ttfamily arXiv:2109.09462}}].

\bibitem{Lu21}
K. Lu,
\emph{A note on odd reflections of super Yangian and Bethe ansatz}, \textit{Lett. Math. Phys.}  (2022), \textbf{112}:29, 26pp,
  [\href{https://arxiv.org/abs/2111.10655}{{\ttfamily arXiv:2111.10655}}].

\bibitem{MV17}
C. Marboe and D. Volin, \textit{Fast analytic solver of rational Bethe equations}, \textit{J. Phys. A} \textbf{50}
(2017) 204002,
  [\href{https://arxiv.org/abs/1608.06504}{{\ttfamily arXiv:1608.06504}}].

\bibitem{GJS22} 
J. Gu, Y. Jiang and M. Sperling,
{\it Rational $Q$-systems, Higgsing and Mirror Symmetry}, {\emph{arXiv e-prints} (Aug., 2022) arXiv:2208.10047},
  [\href{https://arxiv.org/abs/2208.10047}{{\ttfamily arXiv:2208.10047}}].

\bibitem{CLV22}
D. Chernyak, S. Leurent and D. Volin,
{\it Completeness of Wronskian Bethe equations for
rational $\mathfrak{gl}_{\mathfrak m|\mathfrak n}$
spin chains}, \textit{Commun. Math. Phys.} \textbf{391} (2022), 969–1045,
  [\href{https://arxiv.org/abs/2004.02865}{{\ttfamily arXiv:2004.02865}}].

\bibitem{MNV}
J. M. Maillet, G. Niccoli and  L. Vignoli,
{\it Separation of variables bases for integrable $gl_{\mc M|\mc N}$ and Hubbard models}, \textit{SciPost Phys.} \textbf{9}, 060 (2020), 55 pages,
  [\href{https://arxiv.org/abs/1907.08124 }{{\ttfamily arXiv:1907.08124}}].

\bibitem{Ryan}
P. Ryan,
\textit{Integrable systems, separation of variables and the Yang-Baxter equation}, {\emph{arXiv e-prints} (Mon., 2022) arXiv:2201.12057},
  [\href{https://arxiv.org/abs/2201.12057}{{\ttfamily arXiv:2201.12057}}].

\bibitem{Lu20}
K. Lu,
\textit{Perfect integrability and Gaudin models},
\textit{SIGMA} \textbf{16} (2020), 132, 10 pages,
  [\href{https://arxiv.org/abs/2008.06825}{{\ttfamily arXiv:2008.06825}}].

\bibitem{MTV09} 
E. Mukhin, V. Tarasov, and A. Varchenko,
{\it Schubert calculus and representations of general linear group},
J. Amer. Math. Soc. {\bf 22} (2009), no. 4, 909--940,
  [\href{https://arxiv.org/abs/0711.4079}{{\ttfamily arXiv:0711.4079}}].	

\bibitem{MTV14}
E. Mukhin, V. Tarasov, and A. Varchenko, 
{\it Spaces of quasi-exponentials and representations of the Yangian $\mathrm{Y}(\gl_N)$}, Transform. Groups \textbf{19} (2014), no. 3, 861--885,
  [\href{https://arxiv.org/abs/1303.1578}{{\ttfamily arXiv:1303.1578}}].	

\bibitem{Mol07}
A. Molev,
Yangians and Classical Lie Algebras, Math. Surveys Monogr. \textbf{143}, Amer. Math. Soc., Providence, RI, 2007.

\bibitem{Tal06}
D. Talalaev,
{\it The quantum Gaudin system},
Funct. Anal. Appl., \textbf{40} (2006), pp. 73--77,
  [\href{https://arxiv.org/abs/hep-th/0404153}{{\ttfamily hep-th/0404153}}].

\bibitem{CF08}
A. Chervov and G. Falqui,
{\it Manin matrices and Talalaev's formula}, \textit{J. Phys. A} {\bf 41} (2008), 194006,
  [\href{https://arxiv.org/abs/0711.2236}{{\ttfamily arXiv:0711.2236}}].

\bibitem{PRS17}
S. Pakuliak, E. Ragoucy and N. Slavnov,
{\it Bethe vectors for models based on the super-Yangian $Y(\gl(m|n))$},
\textit{J. Integrable Syst.} {\bf 2} (2017), 31pp,
  [\href{https://arxiv.org/abs/1604.02311}{{\ttfamily arXiv:1604.02311}}].

\bibitem{LM21b}
K. Lu and E. Mukhin,
\emph{Jacobi-Trudi identity and Drinfeld functor for super Yangian}, \textit{Int. Math. Res. Not.} \textbf{2021} (2021), no. 21, 16749--16808,
  [\href{https://arxiv.org/abs/2007.15573}{{\ttfamily arXiv:2007.15573}}].

\bibitem{FH:2015}
E.  Frenkel and D.  Hernandez,
{\it Baxter's  Relations  and  Spectra  of  Quantum  Integrable  Models},  \textit{Duke  Math.  J.} \textbf{164} (2015), no. 12, 2407-2460,
  [\href{https://arxiv.org/abs/1308.3444}{{\ttfamily arXiv:1308.3444}}].

\bibitem{FJMM17}
B. Feigin, M. Jimbo, T. Miwa and E. Mukhin,
{\it Finite type modules and Bethe ansatz equations}, \textit{Ann. Henri Poincare} \textbf{18} (2017), no. 8, 2543-2579,
  [\href{https://arxiv.org/abs/1609.05724}{{\ttfamily arXiv:1609.05724}}].

\bibitem{MVY15}
E. Mukhin, B. Vicedo and C. Young, 
\emph{Gaudin models for $\gl(m|n)$}, \textit{J. Math. Phys.} {\bf 56} (2015), no.\,5, 051704, 30pp,
  [\href{https://arxiv.org/abs/1404.3526}{{\ttfamily arXiv:1404.3526}}].

\bibitem{FFR94} 
B. Feigin, E. Frenkel and N. Reshetikhin, {\it Gaudin model, Bethe ansatz and critical level}, \textit{Comm. Math. Phys.} {\bf 166} (1994), no. 1, 27--62,
  [\href{https://arxiv.org/abs/hep-th/9402022}{{\ttfamily hep-th/9402022}}].

\bibitem{Lu22}
K. Lu,
\textit{Completeness of Bethe ansatz for Gaudin models associated with $\mathfrak{gl}(1|1)$}, \textit{Nuclear Phys. B} \textbf{980} (2022), Paper No. 115790, 23 pp,
  [\href{https://arxiv.org/abs/2202.08162}{{\ttfamily arXiv:2202.08162}}].





\end{thebibliography}

\end{document}